\RequirePackage[l2tabu, orthodox]{nag} %
\documentclass[preprint, sort&compress]{elsarticle}
\pdfoutput=1
\usepackage[utf8]{inputenc}
\usepackage[T1]{fontenc}
\usepackage[english]{babel}
\usepackage{mathtools}
\usepackage{amssymb}
\usepackage{amsthm}
\usepackage{tabularx}
\usepackage{multirow}
\usepackage{ebproof}
\usepackage{tikz}
\usetikzlibrary{arrows, calc, positioning, shapes}
\usepackage{xspace}
\usepackage[autostyle = true, english = american, strict = true, autopunct=true]{csquotes}
\usepackage[final, babel]{microtype}
\usepackage{fnpct}
\usepackage{xcolor} %
\usepackage{subcaption}
\definecolor{colorp1}{RGB}{229,245,249}
\definecolor{colorp2}{RGB}{153,216,201}
\definecolor{colorp3}{RGB}{44,162,95}
\mathtoolsset{mathic} %
\usepackage{hyperref}
\hypersetup{
	plainpages=false,
	bookmarksnumbered=true,
	bookmarksopen=true,
	bookmarksdepth=2,
	breaklinks=true,
	linkbordercolor=colorp1,
	urlbordercolor=colorp2,
	citebordercolor=colorp2,
	pdfdisplaydoctitle= true,
	pdfinfo={
 Author={Clément Aubert, Ioana Cristescu},
 Title={Contextual equivalences in configuration structures and reversibility},
 Keywords={Formal semantics, Process algebras and calculi, Reversible CCS, Hereditary history preserving bisimulation, Strong barbed congruence, Contextual caracterisation},
 Creator={PDFLaTeX},
 Producer={PDFLaTeX},
 Lang={en}
	}
}

\theoremstyle{plain}
\newtheorem{theorem}{Theorem}
\newtheorem{lemma}{Lemma}
\newtheorem{proposition}{Proposition}

\newtheorem{corollary}{Corollary}
\newtheorem{conjecture}{Conjecture}

\theoremstyle{definition}
\newtheorem{definition}{Definition}
\newtheorem{remark}{Remark}
\newtheorem{notations}{Notations}
\newtheorem{example}{Example}

\addto\extrasenglish{

}

\newcommand{\fw}{\mathrel{\longrightarrow}} %
\newcommand{\bw}{\mathrel{\rightsquigarrow}} %
\newcommand{\fbw}{\mathrel{\twoheadrightarrow}} %

\newcommand{\red}{\mathrel{\fw}}
\newcommand{\revred}{\mathrel{\bw}}
\newcommand{\redl}[1]{\mathrel{\overset{#1}{\red}}}
\newcommand{\revredl}[1]{\mathrel{\overset{#1}{\revred}}}
\newcommand{\fbl}[1]{\mathrel{\overset{#1}{\fbw}}}

\newcommand{\idmem}[2]{#1 : #2} %
\newcommand{\compfleche}[3]{\overset{\idmem{#1}{#2}}#3}
\newcommand{\fwlts}[2]{\compfleche{#1}{#2}{\fw}}
\newcommand{\bwlts}[2]{\compfleche{#1}{#2}{\bw}}
\newcommand{\fbwlts}[2]{\compfleche{#1}{#2}{\fbw}}

\newcommand{\congru}{\equiv} %
\newcommand{\bisim}{\sim}
\newcommand{\bbisim}{\overset{\cdot}{\bisim}^{\tau}}
\newcommand{\bfbisim}{\overset{\cdot}{\bisim}^{\tau}}
\newcommand{\bfcong}{\bisim^{\tau}}
\newcommand{\hhpb}{\bisim}

\newcommand{\scbisim}{\mathrel{\bisim^{\tau}}}

\newcommand{\sbfbc}{\mathrel{\sim^{\tau}}}
\newcommand{\Forw}[1]{F_{#1}} %
\newcommand{\Backw}[1]{B_{#1}} %

\newcommand{\emptymem}{\varnothing}

\newcommand{\orig}[1]{O_{#1}} %
\newcommand{\fork}{\curlyvee} %
\newcommand{\substs}[1]{[#1]}
\newcommand{\tempty}{\epsilon}%
\newcommand{\mproc}[2]{#1 \rhd #2}
\newcommand{\tick}{\checkmark}
\newcommand{\addfork}{\mathtt{\curlyvee}} %
\newcommand{\context}{C_{\curlyvee}} %

\newcommand{\integer}{\mathbb{N}}
\newcommand{\compa}{{\uparrow}} %
\newcommand{\setst}{~|~} %

\DeclareMathOperator{\ids}{\mathsf{I}}

\DeclareMathOperator{\mathaddress}{ad}
\DeclareMathOperator{\card}{Card}

\newcommand{\funaddress}[3]{\mathaddress_{#1}(#2, #3)}
\newcommand{\encc}[1]{[\![#1]\!]}
\newcommand{\enc}[1]{\encc{#1}} %
\newcommand{\encr}[1]{[\![#1]\!]}
\newcommand{\fn}[1]{\mathrm{fn}(#1)}
\newcommand{\bn}[1]{\mathrm{bn}(#1)}
\newcommand{\nm}[1]{\mathrm{nm}(#1)}
\newcommand{\names}{\mathsf{N}}
\newcommand{\labels}{\mathsf{L}}
\newcommand{\iso}{\cong} %
\newcommand{\out}[1]{\bar{#1}}
\newcommand{\power}{\mathcal{P}}
\newcommand{\conf}{\mathcal{C}}
\newcommand{\confzero}{{\mathbf{0}}}
\newcommand{\rel}{\mathrel{\mathcal{R}}}
\newcommand{\relCCS}{\mathrel{\mathcal{R}_{\text{CCS}}}}
\newcommand{\relCONF}{\mathrel{\mathcal{R}_{\text{Conf}}}}
\newcommand{\mem}[1]{\langle#1\rangle}
\newcommand{\labl}{\ell}
\newcommand{\restr}{\upharpoonleft}
\newcommand{\erase}[1]{\varepsilon(#1)}
\newcommand{\drule}[1]{\textsc{#1}}
\newcommand{\bs}{\backslash}
\newcommand{\Act}{\textrm{Act}}

\DeclareUnicodeCharacter{00A0}{~} %
\newcommand{\st}{s.t.\@\xspace} %
\newcommand*{\ie}{i.e.\@\xspace}
\newcommand{\withoutlog}{w.l.o.g.\@\xspace}
\newcommand*{\resp}{resp.\@\xspace}

\newcommand{\BNFsepa}{\enspace \Arrowvert \enspace}

\makeatletter
\g@addto@macro\@floatboxreset\centering
\makeatother

\begin{document}

\begin{frontmatter}
	\title{Contextual equivalences in configuration structures and reversibility%
		\tnoteref{t1,t2}%
	}
	\tnotetext[t1]{Extended version of a work presented at \href{http://discotec2015.inria.fr/workshops/ice-2015/}{ICE 2015}~\cite{Aubert2015d}. A part of this work appears, with much more contextual material, in the second author's Ph.D Thesis~\cite{Cristescu2015}.}
	\tnotetext[t2]{This work was partly supported by the ANR-14-CE25-0005 \href{http://lipn.univ-paris13.fr/~mazza/Elica/}{ELICA} and the ANR-11-INSE-0007 \href{http://www.pps.univ-paris-diderot.fr/~jkrivine/ANR/REVER/ANR_REVER/Welcome.html}{REVER}.}

	\author[inria,lacl]{Clément Aubert\fnref{fn1}\corref{cor1}}
	\ead{aubertc@appstate.edu}
	\ead[url]{http://lacl.fr/~caubert/}

	\author[pps]{Ioana Cristescu\corref{cor2}}
	\ead{ioana.cristescu@pps.univ-paris-diderot.fr}
	\ead[url]{http://www.pps.univ-paris-diderot.fr/~ioana/}

	\address[inria]{INRIA}
	\address[lacl]{Université Paris-Est, LACL (EA 4219), UPEC, F-94010 Créteil, France}
	\address[pps]{Diderot, Sorbonne Paris Cité, P.P.S., UMR 7126, F-75205 Paris, France}

	\fntext[fn1]{Present address: Department of Computer Science, Appalachian State University, Boone, NC 28608, USA.}
	\cortext[cor1]{Corresponding author}
	\cortext[cor2]{Principal Corresponding author}
	\journal{Journal of Logical and Algebraic Methods in Programming}
	\begin{abstract}
		Contextual equivalence equate terms that have the same observable behaviour in any context.
		A standard contextual equivalence for CCS is the strong barbed congruence.
		Configuration structures are a denotational semantics for processes in which one define equivalences that are more discriminating, \ie that distinguish the denotation of terms equated by barbed congruence.
		Hereditary history preserving bisimulation (HHPB) is such a relation. %
		We define a strong back-and-forth barbed congruence on RCCS, a reversible variant of CCS.
		We show that the relation induced by the back-and-forth congruence on configuration structures is equivalent to HHPB, thus providing a contextual characterization of HHPB.
	\end{abstract}

\end{frontmatter}

\section*{Introduction}
\addcontentsline{toc}{section}{Introduction}
\label{sec:intro}
\subsection*{Reversibility}
Being able to reverse a computation is an important feature of computing systems.
Reversibility is a key aspect in every system that needs to achieve distributed consensus~\cite{Bouge1988} to escape local states where the consensus cannot be found.
In such problems, multiple computing agents have to reach a common solution.
Allowing independent agents to backtrack and explore the solution space enables them to reach a globally accepted state if given enough time and if a common solution exists.
For example, the dining philophers problem~\cite{Hoare1985} requires a backtracking mechanism to prevent deadlocks.
Rewinding a computation step by step is also a common way to debug programs.
In such settings the step by step approach is often more useful than restarting the program from an initial state.

Importantly, the backtracking mechanism can be integrated to the operational semantics of a programming language, instead of adding a tailor-made implementation on top of each program.
A formal model for reversible concurrent systems needs to address two challenges at the same time: (i) how to compute without forgetting and (ii) what is an optimal notion of legitimate backward moves.
Roughly speaking, the first point is about syntax: processes need to carry a memory that keeps track of everything that has been done (and of the choices that have not been made).
Importantly the needed information to backtrack is recorded in a \emph{distributed} fashion instead of using a centralized store, which could be a bottleneck for the computation.
The second point is tied to the choice of the computation's semantics.
In a sequential program, one backtracks computations in the opposite order to the execution.
However, in a concurrent setting, we do not want to undo the actions precisely in the opposite order than the one in which they were executed, as this order may not materialise.
The concurrency relation between actions has to be taken into account.
It can be argued that the most liberal notion of reversibility is the one that just respects causality: an action can be undone precisely after all the actions that causally depend on it have also been undone.
Then an acceptable backward path is \emph{causally consistent} with the forward computation.

There are different accounts of reversible operational semantics, RCCS~\cite{Danos2004,Danos2007} and CCSK~\cite{Phillips2007} being the two main propositions for a reversible CCS.
In these works, reversibility is embedded into a (classical) process calculus.

\subsection*{Causal models}
In interleaving models, the internal relations between different events cannot be observed.
In particular, causality is not treated as a primitive concept.
On the other hand, non-interleaving semantics have a primitive notion of \emph{concurrency} between computation events.
As a consequence one can also derive a \emph{causality} relation, generally defined as the complement of concurrency.
These models are therefore sometimes called \emph{true-concurrent} or \emph{causal} or, if causality is represented as a partial order on events, \emph{partial order} semantics\footnote{Event and configuration structures were introduced to define domains for concurrency~\cite{Nielsen1979}.
	Causal models are thus often, but inaccurately, called \emph{denotational}: a denotational interpretation is supposed to be invariant by reductions, a property that event structures do not have.}.

A causal model is often an alternative representation of an existing interleaving semantics that helps in understanding the relations between computations in the latter.
Usually in such models, sets of events are considered computational states.
Each set, called a \emph{configuration}, represents a reachable state in the run of the process.
The behaviour of a system is encoded as a collection of such sets.
The set inclusion relation between the configurations stands for the possible paths followed by the execution.
Concurrency and causality are derivable from set inclusion.
In their generality, such models are called \emph{configuration structures}~\cite{Glabbeek2009}, they are a syntax-free and causal model that can interpret multiple calculi.

\emph{Stable families}~\cite{Winskel1986} are configuration structures equipped with a set of axioms, that capture the intended behaviour of a CCS process.
Morphisms of stable families capture sub-behaviours of processes and form a category of stable families.
Process combinators correspond then to universal constructions in this category.
The correspondence with CCS is established through an operational semantics defined on stable families, that we abusively name in that context configurations structures as well.

\subsection*{Behavioural equivalence}
Behavioural equivalences are a major motivation in the study of formal semantics.
For instance, one wants to verify that the execution of a program satisfies its expected behaviour, or that binaries obtained from the same source code, but with different compilation techniques, behave the same.
Thus the interesting equivalences equate terms that behave the same.
Moreover the equivalence should be a congruence: two processes are equivalent if they behave similarly in any context.
Loosely speaking it aims at identifying process that have a common external behaviour in any environment.

Equivalences defined on reduction semantics are often hard to prove.
A proof technique in this case is to define a LTS-based equivalence that is equivalent with the reduction-based one and carry the proofs in LTS semantics.

Behavioural equivalences are defined on the operational semantics and thus cannot access the structure of a term.
The observations one do during the execution of a process are called the \emph{observables} of the relation.
For instance one observes whether the process terminates or whether it interacts with the environment~\cite{Honda1995}.

\subsection*{Causality and reversibility}
Causality and reversibility are tightly connected notions~\cite{Nicola1990,Danos2004}.
Causal consistency is a correctness criterion for reversible computations.
Therefore whenever a reversible semantics is proposed, the calculus has to be equipped first with a causal semantics.

\emph{Prime} LTS are known~\cite{Nielsen1981} to generates a prime event structure.
Since a specific reversible LTS~\cite{Phillips2007} is indeed prime, and moreover since the forward and backward reductions correspond to reductions in its causal representation, reversible models and causal ones are easily derivable from each other.

Notably the connection between reversibility and causality is useful to define meaningful reversible equivalences.
Causal equivalences are more discriminating than the traditional operational ones.
However on a reversible operational semantics one define equivalences of the same expressivity.
Causal equivalences have been extensively studied~\cite{Phillips2012,Bednarczyk1991,Baldan2014,Vogler1993}.
Of particular interest is the hereditary history preserving bisimulation, which was shown to correspond to a LTS-based equivalence for a reversible CCS~\cite{Phillips2007}.

\subsection*{Equivalences on configuration structures}
In CCS equivalences are defined only on forward transitions and are therefore inappropriate to study reversible processes. %

A reversible bisimulation~\cite{Lanese2010} is more adapted but it is not contextual.
We introduce a contextual equivalence on RCCS by adapting the notions of contexts and barbs to the reversible setting.
The resulting relation, called \emph{barbed back-and-forth congruence} is defined similarly to the barbed congruence of CCS except that the backwards reductions are also observed.

Configuration structures provide a causal semantics for CCS.
Equivalences on configuration structures are more discriminating than the ones on the operational setting. %
It is possible to move up and down in the lattice, whereas in the operational semantics, only forward transitions have to be simulated.
As an example, consider the processes \(a \mid b\) and \(a.b+b.a\) that are bisimilar in CCS but whose causal relations between events differ.

In particular we are interested in hereditary history preserving bisimulation (HHPB) in \autoref{def:hhpb}, which equates configuration structures that simulate each others' forward and backward moves.
Phillips and Ulidowski~\cite{Phillips2012} showed that the back-and-forth bisimulation corresponds to HHPB, that can be defined in an operational setting thanks to reversibility.
Allowing both forward and backward transitions gives to the operational world the discriminating power of causal models.
We show that HHPB also corresponds to a congruence on RCCS, the barbed back-and-forth congruence.
It is the a contextual characterisation of HHPB which implies a contextual equivalences in configuration structures.

\subsection*{Outline}
We begin by recalling notions on LTS and CCS, as well as their so-called reversible variants (\autoref{sec:syntax}).
RCCS (\autoref{sec:rccs}) is then proven to be a conservative extension of CCS over the traces: their is a strong bisimulation between a reversible process and a \enquote{classical}, memory-less, process (\autoref{lem:corresp_ccs_rccs}).
Lastly, we adapt the usual CCS notions of contexts, barbs, and barbed congruence to RCCS (\autoref{sec:contextual-equiv-rccs}), thus introducing the back-and-forth barbed congruence (\autoref{def:sbfc_rccs}).

We next introduce the interpretation of reversible process on configuration structures (\autoref{sec:semantics}).
We recall the classical definitions (\autoref{sec:st_fam}) as well as the encoding of CCS terms in configuration structures (\autoref{sec:causal_ccs}).
Encoding of RCCS terms is built on top of it (\autoref{sec:stfam_rccs}), and an operational correspondence between reversible processes and their interpretations is proven (\autoref{lem:operational_corresp}).

Finally, we introduce a notion of context for configuration structures (\autoref{sec:context_stfam}) and study the relation induced on configuration structures by the barbed back-and-forth congruence (\autoref{sec:barbed_congruence_stfam}).
In \autoref{sec:inductive_hhpb} we define the hereditary history preserving bisimilarity and provide a characterisation by inductive relations.
Lastly, we show in \autoref{sec:contextual_hhpb} that HHPB is a congruence (\autoref{prop:HHPB_congr}) and that whenever two configuration structures are barbed back-and-forth congruent, they also are hereditary history preserving bisimilar (\autoref{main-thm}).

Our main contribution is proving that barbed congruence in RCCS corresponds to hereditary history preserving bisimulation, which is defined on configuration structures
As a consequence, it provides a contextual characterization of equivalences defined in non-interleaving semantics.

\subsection*{Limitations}
Our work is restrained to processes that forbid \enquote{auto-concurrency} and \enquote{auto-conflict} (\autoref{rem-concur}).
We do not cover recursion, though a treatment of recursion in configuration structures exists~\cite{Winskel1986}.
\enquote{Irreversible} action is a feature of RCCS~\cite{Danos2004} that is absent of our work.

We tried to stick to canonical notations and to remind of common definitions.
However, we consider the reader familiar with the syntax, congruence relation and reduction rules of CCS.
If not, a quick glance at a textbook~\cite{Milner1989} or at lectures notes~\cite{Amadio2014} should help the reader uneasy with them.

\section{Contextual equivalences in reversibility}
\label{sec:syntax}
Reversibility provides an implicit mechanism to undo computations.
Interleaving semantics use a Labeled Transition System (LTS) to represent computations, henceforth refered to as the \emph{forward} LTS.
In a reversible semantics a second LTS is defined that represents the \emph{backward} moves (\autoref{sec:lts_gen}).

RCCS~\cite{Danos2004, Danos2005, Krivine2006} (\autoref{sec:rccs}) is a reversible variant of CCS, that allows computations to \emph{backtrack}, hence introducing the notions of \emph{forward} and \emph{backward} transitions.
Memories attached to processes store the relevant information to eventually do backward steps.
Without this memory, RCCS terms are essentially CCS terms (\autoref{lem:corresp_ccs_rccs}), but their presence forces to be precise when defining contexts and contextual equivalence for the reversible case (\autoref{sec:contextual-equiv-rccs}).

\subsection{(Reversible) labelled transition systems}
\label{sec:lts_gen}
A labelled transition system is a multi-graph where the nodes are called \emph{states} and the edges, \emph{transitions}.
Transitions are labelled by \emph{actions} and may be fired non-deterministically.

\begin{definition}[Labelled Transition System]
	A \emph{labelled transition system} is a tuple \((\to,S, \Act)\) made of a set \(S\) of \emph{states}, a set \(\Act\) of \emph{actions} (or labels) and a relation \(\to\subseteq S\times \Act\times S\).

	For \(s, s'\in S\) and \(a, b \in \Act \) , we write \(s\redl{a}s'\) for \((s,a,s')\in\to\) and \(s\to s'\) if \(s\redl{a}s'\) for some \(a\in \Act\).

	Elements \(t:s\redl{a}s'\) of \(\red\) are called transitions.
	Two transitions, \(t\) and \(t'\) are \emph{composable}, written \(t;t'\), if the target of \(t\) is the source of \(t'\).
	The empty trace is denoted \(\tempty\).
\end{definition}

\begin{definition}[Trace]
	\label{def:trace}
	A \emph{trace}, denoted by \(\sigma:t_1;\dots;t_n\) is a sequence of composable transitions.
	Except for the empty trace, all traces have a source and a target.

	Define \(\red^{\star}\subseteq S\times \Act^{\star}\times S\) the \emph{reachability} relation as follows:
	\[ s\redl{\alpha_1} \dotsb \redl{\alpha_n}s'\iff \begin{multlined}[t]
		\exists t_1,\dots,t_n\text{ and }s_1,\dots,s_{n+1}\text{ such that }\\
		t_i:s_i\redl{\alpha_i}s_{i+1}\text{ and }s_1=s, s_{n+1} = s'.
		\end{multlined}
	\]
	We say in that case that \(s'\) is \emph{reachable} from \(s\), that \(s'\) is a \emph{derivative} of \(s\), and that \(s\) is an \emph{ancestor} of \(s'\).
\end{definition}

\begin{definition}[Reversible LTS]
	\label{def:rlts}
	Given \((\fw,S, \Act)\) and \((\bw,S, \Act)\) two labelled transition systems defined on the same set of states and actions, we define \((\fbw, S, \Act)\) a third LTS by taking \(\fbw = \fw \cup \bw\).
	By convention, a transition \(s \fw t\) is said to be \emph{forward}, whereas a transition \(t \bw s\) is said to be \emph{backward}.
	In \(t \bw s\), \(s\) is an \emph{ancestor} of \(t\).
\end{definition}

A variety of semantically different backtracking mechanisms exists, for instance,
\begin{itemize}
	\item taking \(\bw=\emptyset\) models a language with only irreversible moves,
	\item in a sequential setting, if \(\fw\) draws a tree, taking \(\bw = \{(t, a, s) \setst s \redl{a} t\}\) forces the backward traces to follow exactly the forward execution.
\end{itemize}

In concurrency, backward traces are allowed if their source and target are respectively the target and source of a forward trace.

\subsection{Reversible CCS}
\label{sec:rccs}
A RCCS term, also called a \emph{monitored process}, is a CCS process equipped with a memory.
A \emph{thread} is a CCS term \(P\) guarded by a memory \(m\) and denoted \(\mproc{m}P\).
Processes can be composed of multiple threads.
The memory acts as a stack for the previous computations.
Each entry in the memory is called a \emph{(memory) event} and has a unique identifier.
The forward transitions push events to the memories while the backward moves pop them out.

\begin{definition}[Names, labels and actions]
	We define \(\names=\{a,b,c,\dots\}\) to be the set of \emph{names} and \(\out{\names}=\{\out{a},\out{b},\out{c},\dots\}\) its \emph{co-names}.
	The complement of a (co-)name is given by a bijection \(\out{[\cdot]}:\names \to \out{\names}\), whose inverse is also denoted by \(\out{[\cdot]}\), so that \(\out{\out{a}}=a\).

	A \emph{synchronisation} is a pair of names that complement each other, as \((a,\out{a})\), and that is denoted with the special symbol \(\tau\), whose complement is undefined.

	Actions are labelled using the set \(\labels = \names \cup \out{\names} \cup \{\tau\}\) of (event) labels defined by the following grammar:
	\begin{align}
		\names \cup \out{\names} : \lambda, \pi & \coloneqq a \BNFsepa \out{a} \BNFsepa \hdots \tag{CCS prefixes}               \\
		\labels : \alpha,\beta                  & \coloneqq \tau \BNFsepa a \BNFsepa \out{a} \BNFsepa \hdots \tag{Event labels}
	\end{align}
	As it is common, we will sometimes use \(a\) and \(b\) to range over names, and call the set of names and co-names simply the set of names.
\end{definition}

Transitions in both directions are decorated by the identifier of the associated event.
Identifiers on the (partial) events are used to remember their synchronisation partners.
Thus to combine into a \(\tau\), the transitions need complementary labels and the same identifier.

\paragraph{Grammar}
Consider the following \emph{process constructors}, also called \emph{combinators} or \emph{operators}:
\begin{align}
	e & \coloneqq \mem{i, \alpha, P} & \tag{memory events} \\
	m &\coloneqq \emptymem \BNFsepa \fork . m \BNFsepa e. m \tag{memory stacks}\\
	P,Q &\coloneqq \lambda.P \BNFsepa P\mid Q \BNFsepa \lambda.P+\pi.Q \BNFsepa P\bs a \BNFsepa 0 \tag{CCS processes}\\
	R, S & \coloneqq m \rhd P \BNFsepa R \mid S \BNFsepa R\bs a \tag{RCCS processes}
\end{align}
A (memory) event \(e=\mem{i, \alpha, P}\) is made of:
\begin{itemize}
	\item An event identifier \(i \in \ids\) that \emph{tags} transitions.
	      We may think of them as \texttt{pid}, in the sense that they are a centrally distributed identifier attached to each transition.
	\item A label \(\alpha\) that marks which action has been fired (in the case of a forward transition), or what action should be restored (in the case of a backward move).
	\item A backup of the whole process \(P\) that has been erased when firing a sum, or \(0\) otherwise.
\end{itemize}

In the memory stack, the fork symbol \(\fork\) marks a parallel composition.
The memory is then copied in two local memories, as despicted in the congruence rule called \enquote{\ref{distrib-memory}} in \autoref{def:congru_rccs}).

Lastly the null process, denoted \(0\), cannot perform any transition.
We will often omit it, so for example we write \(a\mid b\) instead of \(a.0\mid b.0\).

\begin{notations}
	\label{notation:ids}
	\begin{itemize}
		\item We use \(\integer\) for the set of \emph{event identifiers} \(\ids\) and let \(i,j,k\) range over elements of \(\ids\).
		      Forward and backward transitions will be tagged with such identifiers, and so we write \(\fwlts{i}{\alpha}\) and \(\bwlts{i}{\alpha}\).
		      We use \(\fbwlts{i}{\alpha}\) as a wildcard for \(\fwlts{i}{\alpha}\) or \(\bwlts{i}{\alpha}\), and if there are indices \(i_1, \hdots, i_n\) and labels \(\alpha_1, \hdots, \alpha_n\) such that \(R_1 \fbwlts{i_1}{\alpha_1} \dotsb \fbwlts{i_n}{\alpha_n} R_n\), then we write \(R_1 \fbw^{\star} R_n\).
		      We sometimes omit the identifier or the label in the transition.
		\item For \(R\) a reversible process and \(m\) a memory, we denote \(\ids(m)\) (\resp \(\ids(R)\)) the set of identifiers occurring in \(m\) (\resp in \(R\)).
		\item The sets \(\nm{R}\) of names in \(R\), \(\fn{R}\) of free names in \(R\) and \(\bn{R}=\nm{R}\setminus{\fn{R}}\) of bound (or private) names in \(R\) are defined by extending the definition of free names on CCS terms to memories and RCCS terms:
		      \begin{align*}
		      	              &
		      	\begin{aligned}
		      	\fn{P\bs a}   & =\fn{P}\setminus{\{a\}}         \\
		      	\fn{a.P}      & =\fn{\out{a}.P}=\{a\}\cup\fn{P} \\
		      	\fn{P\mid Q}  & =\fn{P+Q}=\fn{P}\cup\fn{Q}      \\
		      	\fn{0}        & =\emptyset
		      	\end{aligned}
		      	\tag{CCS rules}\\
		      	\\&
		      	\begin{aligned}
		      	\fn{R\bs a}   & =\fn{R}\setminus{\{a\}}         \\
		      	\fn{R\mid S}  & =\fn{R}\cup\fn{S}               \\
		      	\fn{m \rhd P} & =\fn{P}
		      	\end{aligned}
		      	\tag{RCCS rules}
		      \end{align*}
	\end{itemize}
\end{notations}

\begin{remark}[On recording the past]
	\label{rm:past}
	To store the information needed to backtrack, RCCS attaches local memories to each thread.
	CCSK~\cite{Phillips2007}, a variant of CCS, simulates reductions by movings a pointer in the term, that is left unchanged.
	Reversible higher-order \(\pi\)~\cite{Lanese2010} uses a centralised, global memory to store the process before a reduction.
	Keys are associated to each reduction, thus reverting a transition with key \(k\) consists in restoring the process associated to \(k\) from the global memory.
	The exact mechanism used for recording does not have an impact on the theory except for the structural rules, as we note in \autoref{rk:congr}.
\end{remark}

The labelled transition system for RCCS is given by the rules of \autoref{fig:lts_rccs}.

\begin{figure}
	\begin{minipage}[b]{.99\linewidth}
		\centering
		\begin{tabular}{c}
			                                                                                                                                                        \\[1em]
			\begin{prooftree}
			\Infer[left label = {\drule{In\(+\)}}, right label = {\(i\notin\ids(m)\)}]{0}{\mproc{m}{a.P + Q}\fwlts{i}{a} \mproc{\mem{i,a,Q}.m}P}
			\end{prooftree}
			                                                                                                                                                        \\[2em]
			\begin{prooftree}
			\Infer[left label = {\drule{Out\(+\)}}, right label = {\(i\notin\ids(m)\)}]{0}{\mproc{m}{\out{a}.P + Q}\fwlts{i}{\out{a}} \mproc{\mem{i,\out{a},Q}.m}P}
			\end{prooftree}
			                                                                                                                                                        \\[2em]
			\begin{prooftree}
			\Infer[left label = {\drule{In\(-\)}}, right label = {\(i\notin\ids(m)\)}]{0}{\mproc{\mem{i,a,Q}.m}P\bwlts{i}{a} \mproc{m}{a.P + Q}}
			\end{prooftree}
			                                                                                                                                                        \\[2em]
			\begin{prooftree}
			\Infer[left label = {\drule{Out\(-\)}}, right label = {\(i\notin\ids(m)\)}]{0}{\mproc{\mem{i,\out{a},Q}.m}P\bwlts{i}{\out{a}} \mproc{m}{\out{a}.P + Q}}
			\end{prooftree}
			                                                                                                                                                        \\[2em]
		\end{tabular}
		\subcaption{Prefix and sum rules}\label{fig:lts_a}
	\end{minipage}

	\vspace{2em}

	\begin{minipage}[b]{.99\linewidth}
		\centering
		\begin{tabular}{c c}
			\\[1em]
			\begin{prooftree}
			\Hypo{R \fwlts{i}{\alpha} R' \quad S \fwlts{i}{\out{\alpha}} S'}
			\Infer[left label = {\drule{Com\(+\)}}]{1}{R \mid S \fwlts{i}{\tau} R' \mid S'}
			\end{prooftree}
			  &
			\begin{prooftree}
			\Hypo{R \fbwlts{i}{\alpha} R'}
			\Infer[left label = {\drule{ParL}}, right label = {\(i\notin\ids(S)\)}]{1}{R \mid S \fbwlts{i}{\alpha} R' \mid S}
			\end{prooftree}
			\\[2em]

			\begin{prooftree}
			\Hypo{R \bwlts{i}{\alpha} R' \quad S \bwlts{i}{\out{\alpha}} S'}
			\Infer[left label = {\drule{Com\(-\)}}]{1}{R \mid S \bwlts{i}{\tau} R' \mid S'}
			\end{prooftree}
			  &
			\begin{prooftree}
			\Hypo{R \fbwlts{i}{\alpha} R'}
			\Infer[left label = {\drule{ParR}}, right label = {\(i\notin\ids(S)\)}]{1}{S \mid R \fbwlts{i}{\alpha} S \mid R'}
			\end{prooftree}
			\\[2em]
		\end{tabular}
		\subcaption{Parallel constructions}\label{fig:lts_c}
	\end{minipage}

	\vspace{2em}

	\begin{minipage}[b]{.99\linewidth}
		\centering
		\begin{tabular}{c c}
			\\[1em]
			\begin{prooftree}
			\Hypo{R \fbwlts{i}{\alpha} R'}
			\Infer[left label = {\drule{Hide}}, right label = {\(a \notin \{\alpha, \out{\alpha}\}\)}]{1}{ R\bs a \fbwlts{i}{\alpha} R'\bs a}
			\end{prooftree}
			  &
			\begin{prooftree}
			\Hypo{R \congru R' \fbwlts{i}{\alpha} S' \congru S}
			\Infer[left label = {\drule{Congr}}]{1}{R \fbwlts{i}{\alpha} S}
			\end{prooftree}
			\\[2em]
		\end{tabular}
		\subcaption{Hiding and congruence}\label{fig:lts_b}
	\end{minipage}
	\caption{Rules of the RCCS LTS}
	\label{fig:lts_rccs}
\end{figure}

The prefix constructor \(a.P\) stands for sequential composition, the process interacts on \(a\) before continuing with \(P\).
Rules \drule{In\(+\)} (for the input) and \drule{Out\(+\)} (for the output) consumes a prefix by adding in the memory the corresponding event.
The backward moves, described by the rules \drule{In\(-\)} and \drule{Out\(-\)}, remove an event at the top of a memory and restores the prefix and the non-deterministic sum.
Those rules are presented with a (guarded) sum, but we consider for instance \(\emptymem \rhd a.P \fwlts{1}{a} \mem{1,a,0}. \emptymem \rhd P\) to be a legal transition, taking \(P + 0\) (which is not syntactically correct) to be \(P\).

Parallel composition \(P\mid Q\) employ the four rules of \autoref{fig:lts_c} to derive a transition.
Rules \drule{Com\(+\)} and \drule{Com\(-\)} depicts two process agreeing to synchronize or to undo a synchronization by providing two dual prefixes\footnote{Notice that since the complement of \(\tau\) is not defined, only inputs and outputs synchronize.}, agreeing on the event identifier and triggering the transitions simultaneously.
Rules \drule{ParL} and \drule{ParR} allow respectively the left or the right process to compute independently of the rest of the process.
In those two later rules, the side condition \(i \notin \ids(S)\) ensures, in the forward direction, the uniqueness of the event identifiers and it prevents, in the backward direction, a part of a previous synchronisation to backtrack alone.

Once the name \(a\) is \enquote{hidden in \(P\)}, that is, made private to the process \(P\), it cannot be used to interact with the environment.
This situation is denoted with \(P\bs a\) and illustrated in rule \drule{Hide}.
Whenever the private name \(a\) is encountered in the environment, \(\alpha\)-renaming of \(a\) is done inside \(P\):
\[P\bs a =_{\alpha} (P\substs{b/a})\bs b\]
where \(P\substs{b/a}\) stands for process \(P\) in which \(b\) substitutes \(a\).
We say that the hiding operator is a binder for the private name \(a\).

The structural congruence, whose definition follows, is applied on terms by the rule \drule{Congr}.
It is built on top of some of the corresponding rules for CCS, and rewrites the terms under the memory or distributes it between two forking processes.

\begin{definition}[Structural congruence]
	\label{def:congru_rccs}
	Structural congruence on monitored processes is the smallest equivalence relation up to uniform renaming of identifiers generated by the:
	\begin{align}
		m \rhd (P + Q)       & \congru m \rhd (Q + P)                                                                              \\
		m \rhd ((P + Q) + R) & \congru m \rhd (P + (Q + R))                                                                        \\
		                     & \frac{P =_\alpha Q}{m \rhd P \congru m \rhd Q} \tag{\(\alpha\)-conversion}                          \\
		m \rhd (P \mid Q)    & \congru (\fork . m \rhd P \mid \fork . m \rhd Q)\tag{distribution memory}\label{distrib-memory} %
	\end{align}
	Adding a \emph{scope of restriction} rule \(m \rhd P\bs a \congru (m \rhd P)\bs a\) could be done at the price of a cumbersome definition of what a free name in a memory is.
\end{definition}

\begin{remark}[On reduction semantics]
	\label{rk:congr}
	Correctness criteria for reversible semantics mostly relate it with its \emph{only-forward} counterpart.
	However one may be interested in defining a reduction semantics for the LTS of \autoref{fig:lts_rccs} if only to relate RCCS with other reversible semantics for CCS.
	One, then needs a congruence relation on RCCS terms that has the monoid structure for parallel composition and the null process \(0\).
	However, due to the fork constructor, the associativity does not hold:
	\[(R_1\vert R_2)\vert R_3\not\congru R_1\vert (R_2\vert R_3).\]
	Other reversible calculi, in particular the reversible higher-order \(\pi\)-calculus~\cite{Lanese2010} fares better: its congruence relation respects associativity, thanks to a mechanism that uses bounded names for forking processes.
	Then \(\alpha\)-renaming can be applied on these forking names.

	Alternatively, one could use \enquote{at distance rewriting}~\cite{Accatoli2013} to bypass the lack of flexibility of our structural congruence.
\end{remark}

In RCCS not all syntactically correct processes have an operational meaning.
Consider for instance
\[\fork . \mem{i, \alpha, 0} . \emptymem \triangleright P \mid \emptymem \triangleright Q.\]
To make a backward transition, one should first apply the congruence rule called \enquote{\ref{distrib-memory}} and then look for a rule of the LTS to apply.
But this is impossible, since the memory on the right-hand side of the parallel operator does not contain a fork symbol (\(\fork\)) at its top.
The distributed memory does not agree on a common past, blocking the execution, but this term is correct from a syntactical point of view.
In the following, we will consider only the semantically correct terms, called \emph{coherent}.

\begin{definition}[Coherent process]
	\label{def:coherent_rccs}
	A RCCS process \(R\) is \emph{coherent} if there exists a CCS process \(P\) such that \(\emptymem \rhd P \fw^{\star} R\).
\end{definition}

Coherent terms are also called \emph{reachable}, as they are obtained by a forward execution from a process with an empty memory.
Coherence of terms can equivalently be defined in terms of coherence on memories~\cite[Definition~1]{Danos2004}.

Backtracking is non-deterministic because backtracking is possible on different threads.
However, it is noetherian and confluent as backward synchronisations are deterministic~\cite[Lemma~11]{Danos2005}.

\begin{lemma}[Unique origin]
	If \(R\) is a coherent process, then \(\forall R'\) such that \(R \congru R'\) or \(R \fbw R'\), then \(R'\) is also coherent.
	Up to structural congruence, there exists a unique process \(P\) such that \(R \bw^{\star} \emptyset \triangleright P\), we call it \emph{the origin of \(R\)} and denote it \(\orig{R}\).
\end{lemma}

Lastly, we recall a useful result, that asserts that every reversible trace can be rearranged as a sequence of only-backward moves followed by a sequence of only-forward moves.

\begin{lemma}[Parabolic traces, {\cite[Lemma 10]{Danos2004}}]
	\label{lem:rearrange_trace}
	If $R \fbw \cdots \fbw S$, then there exists \(R'\) such that \(R \bw^{\star} R'\fw^{\star} S\).
\end{lemma}

It is natural to wonder if our reversible syntax is a conservative extension of CCS.
We will make sure in the following that the forward rules in the reversible LTS correspond to the LTS of the natural semantics.

\begin{definition}[Map from RCCS to CCS]
	\label{def:rccs_ccs}
	We define inductively a map \(\erase{\cdot}\) from RCCS terms to CCS terms by erasing the memory:
	\begin{align*}
		\erase{m\rhd P} & =P &   &   & \erase{R\vert S} & =\erase{R}\vert\erase{S} &   &   & \erase{R\bs a} & =(\erase{R})\bs a
	\end{align*}
\end{definition}

In the following lemma, we denote \(\redl{\alpha}\) the standard rewriting rule on CCS terms.
\begin{lemma}[Strong \enquote{forward} bisimulation between \(R\) and \(\erase{R}\)]
	\label{lem:corresp_ccs_rccs}

	For all \(R\) and \(S\), \(R \fwlts{i}{\alpha} S\) for some \(i\) iff \( \erase{R} \redl{\alpha} \erase{S} \).
\end{lemma}

\subsection{Contextual equivalences}
\label{sec:contextual-equiv-rccs}

Contextual equivalence for CCS terms~\cite{Milner1992} is now standard, but its extension to RCCS is not straightforward, since contexts needs to be properly defined (\autoref{def:rccs_context}).
As usual, reductions are part of the observables, but observing only them results in a too coarse relation, and adding termination is not relevant in concurrency.
\emph{Barbed congruence} (\autoref{def:barbed_congr}) has proven to be the right notion for CCS, and we revisit it for RCCS terms.
We begin by recalling definitions of context and observables for CCS.

\begin{definition}[Context]
	\label{def:context}
	A context is a process with a hole \([\cdot]\) defined formally by the grammar:
	\[C[\cdot]\coloneqq [\cdot] \BNFsepa \lambda.C[\cdot] \BNFsepa P \mid C[\cdot] \BNFsepa C[\cdot]\bs a\]
\end{definition}

\begin{definition}[Barbs]
	\label{def:barb}
	Write \(P \downarrow_\alpha\) if there exists \(P'\) such that \(P \redl{\alpha} P'\).
\end{definition}

We now define a contextual equivalence where reductions and barbs are the observables.

\begin{definition}[Barbed congruence]
	\label{def:barbed_congr}
	The \emph{barbed bisimulation} is a symmetric relation \({\rel}\) on CCS processes such that whenever \(P \rel Q\) the following holds:
	\begin{align}
		P\red P'      & \implies \exists Q' \text{ \st~} Q\red Q'\text{ and }Q\rel Q'\tag{closed by reduction} \\
		P\downarrow a & \implies Q\downarrow a\tag{barb preserving}
	\end{align}
	If there exists a barbed bisimulation between \(P\) and \(Q\) we write \(P \bbisim Q\) and say that \(P\) and \(Q\) are \emph{barbed bisimilar}.

	If \(\forall C[\cdot]\), \(C[P] \bbisim C[Q]\), we write \(P \scbisim Q\) and say that \(P\) and \(Q\) are \emph{barbed congruent}.
\end{definition}

An interesting proposition allows to restrict the grammar of contexts in the following.

\begin{proposition}
	\label{prop:context_that_counts}
	\(\forall a, P_1, P_2 , Q, \lambda, a\),
	\[
		P_1 \bbisim P_2 \implies
		\begin{cases}
			\lambda.P_1 \bbisim \lambda.P_2 \\
			P_1 \bs a\bbisim P_2 \bs a      \\
			P_1 + Q\bbisim P_2 + Q
		\end{cases}
	\]
\end{proposition}

\begin{proof}
	\begin{enumerate}
		\item \(P_1 \bbisim P_2 \implies \lambda.P_1 \bbisim \lambda.P_2\).
		      From CCS's grammar, \(\lambda\neq\tau\), hence \(\nexists P_1', P_2'\) such that \(P \redl{\tau} P_1'\) or \(P_2 \redl{\tau} P_2'\).
		      The relation \(\{\lambda.P_1,\lambda.P_2\}\) is trivially a barbed bisimulation.
		\item \(P_1 \bbisim P_2 \implies P_1 \bs a \bbisim P_2 \bs a\).
		      Let us denote \({\rel_1}\) the largest barbed bisimulation for \(P_1\) and \(P_2\).
		      We show that the relation \({\rel_2}=\{P_1 \bs a, P_2 \bs a \setst P_1 \rel_1 P_2\}\) is a barbed bisimulation.
		      We have to show that:
		      \begin{itemize}
		      	\item \(\forall b\) such that \(P_1\bs a\downarrow \beta\) then \(P_2\bs a\downarrow \beta\).
		      	      It follows from \(P_1\bs a\downarrow \beta\implies P_1 \downarrow \beta\) and \(\beta\neq a\).
		      	\item \(P_1\bs a \redl{\tau} P_1'\) implies that \(P_2\bs a\redl{\tau} P_2'\) and \(P_1' \rel_2 P_2'\).
		      	      By structural induction on the transition \(P_1\bs a \redl{\tau} P_1'\) we have that \(\exists P_1''\) such that
		      	      \(P_1\redl{\tau} P_1''\) and \(P_1''\bs a=P_1'\).
		      	      As \(P_1 \rel_1 P_2\) there exists \(P_2''\) such that \(P_2 \redl{\tau}P_2''\) and we apply the rule \drule{Hide} we get \(P_2\bs a \redl{\tau} P_2''\bs a\).
		      	      Thus there exists \(P_2'=P_2''\bs a\) and \(P_1' \rel_2 P_2'\).
		      \end{itemize}
		      It follows similarly for the barbs and reductions on \(P_2\).

		\item \(P_1 \bbisim P_2 \implies P_1 + Q \bbisim P_2 + Q\).
		      Let us denote \({\rel_1}\) the largest barbed bisimulation for \(P_1\) and \(P_2\).
		      We show that the relation \({\rel_2}=\rel_1\cup\{P_1, P_2\}\) is a barbed bisimulation.
		      As above, we show that:
		      \begin{itemize}
		      	\item \(\forall\alpha\) such that \((P_1 + Q)\downarrow \alpha\) then \((P_2 + Q)\downarrow \alpha\).
		      	      It follows from the subcases \(P_1 \downarrow \alpha\) (hence \(P_2 \downarrow \alpha\)) or \(Q \downarrow \alpha\).
		      	\item \(P_1 + Q \redl{\tau} P_1'\).
		      	      From rules \drule{SumL} and \drule{SumR} either \(Q \redl{\tau} P_1'\) or \(P_1 \redl{\tau} P_1'\).
		      	      In the first case we deduce, using rule \drule{SumL} that \(P_1 + Q\redl{\tau}P_1'\), and therefore \(P_1' \rel_2 P_1'\).
		      	      In the second case, we apply rule \drule{SumR} and have that \(P_2 \redl{\tau}P_2'\) and \(P_1' \rel_1 P_2'\), which concludes our proof.
		      \end{itemize}
		      It follows similarly for the barbs and reductions on \(P_2\).\qedhere
	\end{enumerate}
\end{proof}

\begin{corollary}
	If a context \(C[\cdot]\) does not contain the parallel operator, then for all \(P\), \(Q\), \(C[P] \not\bbisim C[Q]\) implies \(P \not\bbisim Q\)
\end{corollary}

Stated differently, this implies that discriminating contexts regarding barbed congruence involve parallel composition.
As we will focus on this relation, we will only consider in the following the contexts to be parallel compositions:
\[C[\cdot]\coloneqq [\cdot] \BNFsepa P \mid [\cdot] \]

This is handy to define RCCS context, but some subtleties remain.
A context has to become an executable process regardless of the process instantiated with it.
We say that a context has a coherent memory, it may backtrack up to the context with an empty memory (similar to the \autoref{def:coherent_rccs} of coherent processes).
We distinguish three types of contexts, depending on their memories:
\begin{itemize}
	\item Contexts with an empty memory.
	\item Contexts with a non-empty memory that is coherent on its own\footnote{Up to minor addition of \(\fork\) symbols, as explained later on.}.
	      The process that we instantiate with it can be
	      \begin{itemize}
	      	\item incoherent\footnotemark[\value{footnote}], in which case we conjecture that the term obtained after instantiation is also incoherent,
	      	\item coherent on their own\footnotemark[\value{footnote}], in which case it is possible to backtrack the memory of the context up to the empty memory.
	      \end{itemize}
\end{itemize}
Hence \withoutlog we consider contexts without memory and contexts with coherent memories to be equivalent.
These are the types of contexts that we use throughtout the article.
However, a third case remains:
\begin{itemize}
	\item Contexts that have a non-coherent memory.
	      There exists incoherent terms whose instantiation with an incoherent context is coherent.
	      For instance, \(C=\mem{1,a,0}.\fork . \emptymem \rhd P\mid [\cdot]\) and \(R=\mem{1,\out{a},0}.\fork . \emptymem \rhd P'\) are incoherent individually, but \(C[R]\) is coherent and can backtrack to \(\fork . \emptymem \rhd P \mid \fork . \emptymem \rhd P'\).
	      We leave this case as future work.
\end{itemize}

The \enquote{up to minor addition of \(\fork\) symbols} comes from a simple consideration on the parallel composition in RCCS.
A process with an empty memory compose with a RCCS term if fork symbols are added to reflect the parallel composition.
For instance, two processes with an empty memory \(\emptymem \rhd P\) and \(\emptymem \rhd P'\) compose and we obtain
\[\fork.\emptymem \rhd P\mid \fork.\emptymem \rhd P'\congru \emptymem \rhd (P\mid P')\]
instead of \(\emptymem \rhd P\mid\emptymem \rhd P'\), an incoherent process.

We define a rewriting function on RCCS processes, that adds a fork symbol at the beginning of a memory.
It allows a process with a memory to compose with a context.

\begin{definition}[RCCS context]
	\label{def:rccs_context}
	Define \(\addfork(R)\) the operator that adds a fork symbol at the beginning of the memory of each thread in \(R\):
	\begin{align*}
		\addfork(R_1\mid R_2) & = \addfork(R_1)\mid\addfork(R_2)                       \\
		\addfork(R\bs a)      & = (\addfork(R))\bs a                                   \\
		\addfork(m\rhd P)     & = m'.\fork.\emptymem\rhd P\text{ where }m=m'.\emptymem \\
		\addfork(0)           & = 0
	\end{align*}

	Define \(\context[R]\) as follows
	\[
		\context[R] =
		\begin{cases}
			R                                       & \text{if } C[\cdot]=[\cdot]          \\
			\fork.\emptymem \rhd P \mid \addfork(R) & \text{if } C[\cdot] = P \mid [\cdot]
		\end{cases}
	\]
\end{definition}

RCCS context are basically CCS context with additional fork symbols in the memory of the context and in the memory of the process instantiated.
We now verify that \(\context[R]\) is a coherent process, using the function \(\erase{\cdot}\) that erases the memories from a term (\autoref{def:rccs_ccs}).

\begin{proposition}
	For all \(R\) and \(\context[\cdot]\), \(\emptymem\rhd C[\erase{\orig{R}}]\fw^{\star} \context[R]\).
\end{proposition}
\begin{proof}
	Let \(C[\cdot] = P \mid [\cdot]\) and \(\orig{R}= \emptymem\rhd P'\).
	By definition and application of the congruence rules, we have that
	\begin{align*}
		\emptymem\rhd C[\erase{\orig{R}}] & = \emptymem\rhd \big(P\mid \erase{\orig{R}}\big)              \\
		                                  & = \emptymem\rhd \big(P\mid P'\big)                            \\
		                                  & \congru (\fork.\emptymem\rhd P) \mid (\fork.\emptymem\rhd P')
	\end{align*}
	We have from the trace \(\orig{R}\red^{\star} R\) that
	\[
		(\fork.\emptymem\rhd P) \mid (\fork.\emptymem\rhd P') \red^{\star} (\fork.\emptymem\rhd P) \mid \addfork (R) = \context[R].	\qedhere
	\]
\end{proof}

\begin{example}
	Let \(R=\fork.m . \emptymem \rhd P_1\mid \fork.m . \emptymem \rhd P_2\) and \(C[\cdot] = P \mid [\cdot]\).
	Let us rewind \(R\) to its origin:
	\begin{align*}
		R & =\fork.m . \emptymem \rhd P_1\mid \fork.m . \emptymem \rhd P_2 \\
		  & \qquad \congru m . \emptymem \rhd (P_1\mid P_2)                \\
		  & \qquad \qquad \bw^{\star} \emptymem \rhd P'                    \\
		  & \qquad \qquad \qquad = \orig{R}
	\end{align*}
	We instantiate the context with \(\orig{R}\) and redo the execution from the origin of \(R\) up to \(R\):
	\begin{align*}
		\context[\orig{R}] & = (\fork.\emptymem\rhd P)\mid (\fork.\emptymem \rhd P')                                                                \\
		                   & \qquad \fw^{\star}(\fork.\emptymem\rhd P)\mid \big( (m.\fork.\emptymem \rhd P_1) \mid (m.\fork.\emptymem\rhd P_2)\big) \\
		                   & \qquad \qquad =(\fork.\emptymem\rhd P)\mid\addfork(R)                                                                  \\
		                   & \qquad \qquad \qquad= \context[R]
	\end{align*}
	Hence we have that \(\context[\orig{R}]\fw^{\star}\context[R]\).
\end{example}

Once this delicate notion of context for reversible process is settled, extending the CCS barbs (\autoref{def:barb}) and barbed congruence (\autoref{def:barbed_congr}) are straightforward.

\begin{definition}[RCCS barbs]
	\label{def:barb_rccs}
	We write \(R \downarrow_\alpha\) if there exists \(i \in I\) and \(R'\) such that \(R \fwlts{i}{\alpha} R'\).
\end{definition}

\begin{definition}[Back-and-forth barbed congruence]
	\label{def:sbfc_rccs}
	A \emph{back-and-forth bisimulation} is a symmetric relation on coherent processes \(\rel\) such that if \( R \rel S\), then
	\begin{align*}
		R \bwlts{i}{\tau} R' & \implies \exists S' \text{ \st~} S \bwlts{i}{\tau} S' \text{ and } R' \rel S'; \tag{back} \label{arriere} \\
		R\fwlts{i}{\tau} R'  & \implies \exists S' \text{ \st~} S \fwlts{i}{\tau} S'\text{ and } R' \rel S'; \tag{forth} \label{avant}
		\intertext{and it is a \emph{back-and-forth \emph{barbed} bisimulation} if, additionally,}
		R \downarrow_{a}     & \implies S \downarrow_{a}. \tag{barbed} \label{barbe}
	\end{align*}
	We write \( R \bfbisim S\) if there exists \(\rel\) a back-and-forth barbed bisimulation such that \( R\rel S\).

	The \emph{back-and-forth barbed congruence}, denoted \( R \bfcong S\), holds if for all context \(\context[\cdot]\), \(\context[\orig{R}] \bfbisim \context[\orig S]\).
\end{definition}

From the definition of \(R \bfcong S\), the following lemma trivially holds.
\begin{lemma}
	\label{lem-orig-sbfbc}
	For all \(R\) and \(S\), \(R \sbfbc S \implies \orig{R} \sbfbc \orig{S}\).
\end{lemma}

However, the converse does not hold as \(R\) and \(S\) can be any derivative of the same origin, as illustrated below.

\begin{example}
	Let \(R = \mem{1, a, b.Q} . \emptyset \rhd P\) and \(S = \mem{2, b, a.P} . \emptyset \rhd Q\), with \(P \not\bfbisim Q\).
	We have that \(\orig{R} \bfcong \orig{S}\), as \(\orig{R} = \orig{S}\), but \(R \not\bfcong S\), as \(P \not\bfbisim Q\):
	\begin{center}
		\begin{tikzpicture} %
			\node (or) at (1,1) {\(\orig{R}\sbfbc \orig S\)};
			\node (r) at (1, -0.5) {\(\mem{1, a, b.Q} . \emptyset \rhd P \not\sbfbc \mem{2, b, a.P} . \emptyset \rhd Q\)};
			\draw [->] (or) -- node[left]{\small{\(1:a\)}~~} (-0.8,-0.3);
			\draw [->] (or) -- node[right]{~\small{\(2:b\)}} (2.8,-0.3);
		\end{tikzpicture}
	\end{center}
\end{example}

Note that even if the context is defined for any reversible process, we instantiate the context with processes with an empty memory in \autoref{def:sbfc_rccs}.
If instead we had defined \( R \bfcong S\) iff for all contexts, there exists \(\rel\) such that \(\context[R] \rel \context[S]\), then \autoref{lem-orig-sbfbc} would not hold.
We highlight this in the following example.

\begin{example}
	Let us consider the processes \(\emptymem\rhd a.P + Q\) and \(\emptymem \rhd a.P\) which can do a transition on \(a\) to obtain \(R=\mem{1,a, Q}\rhd P\) and \(S=\mem{1,a}\rhd P\).
	We have that \(R\) and \(S\) are back-and-forth barbed bisimular.
	As we are using contexts without memory, there is no context able to backtrack on \(a\).
	\begin{center}
		\begin{tikzpicture} %
			\node (or) at (-1,1) {\(\orig{R}=\emptymem\rhd a.P+Q\)};
			\node (rel1) at (1,1) {\(\not\bfbisim\)};
			\node (os) at (2.5,1) {\(\orig S =\emptymem\rhd a.P\)};
			\node (r) at (-1,0) {\(R =\mem{1,a,Q}\rhd P\)};
			\node (rel2) at (1,0) {\(\sbfbc\)};
			\node (s) at (2.5,0) {\(S =\mem{1,a}\rhd P\)};
			\draw [->] (or) -- node[right]{\small{\(1:a\)}} (r);
			\draw [->] (os) -- node[right]{\small{\(1:a\)}} (s);
		\end{tikzpicture}
	\end{center}
\end{example}

\begin{remark}[On backward barbs]
	\label{rk:back_barb}
	Let us informally argue that backward barbs are not an interesting addition to a contextual equivalence.
	One can always define ad-hoc barbs that potentially change the equivalence relations, however we end up with relations that have no practical meaning.
	We consider below another definition of barb~\cite[Definition 2.1.3]{Madiot2015}, which gives an intuitive reading and is not syntax-specific.

	Let the tick (\(\tick\notin \names\)) be a special symbol denoting termination.
	A \emph{barb} is an interaction with a context that can do a tick immediately after:
	\[P\downarrow_\alpha \iff P\mid \out{\alpha}.\tick\redl{\tau}Q\mid\tick \text{ for some }Q.\]

	Note that the definition above implies that (i) the barb is an interaction with a context that can terminate immediately after and (ii) the interaction \emph{blocks} the termination on the context side, \ie no further transition is possible on that side.

	If we try to apply this definition to a backward barb then the tick has to be in the memory of the context and blocked by another action, \ie the context has to be of the form \(C[\cdot] \congru [\cdot] \mid (\mem{i,\out{\alpha},0}.\mem{\tick}.\emptymem \rhd 0) \).
	This raises multiple problems:
	\begin{enumerate}
		\item Syntactically, \(\tick\) becomes a prefix, rather than a \enquote{terminal process}, \ie terms of the form \(\tick . a . P\) appear.
		      This contradict the intuition that this symbol stands for termination.
		\item In a situation where \(C[R] \bwlts{i}{\tau} R' \mid \mem{\tick}.\emptymem \rhd 0 \), the \(\tick\) symbol is not observable, and \(R'\) could continue its computation before \(\tick\) is popped from the context's memory.
		      So we would have to add the content of the memory to what is observable.
		      But in that case, one might as well look directly in the process' memory if a label is present.
		\item Lastly, defining the backward barb as the capability to do a backward step, and having immediately after the forward barb, seems to be equivalent to any reasonable definition of backward barb.
	\end{enumerate}

	Thus we argue that the backward barbs are a contrived notion.
\end{remark}

\section{Configuration structures as a model of reversibility}
\label{sec:semantics}
Causal models take causality and concurrency between events as primitives.
In configuration structures, configurations stands for computational states and the set inclusion represents the executions, so that in each state one can infer a local order on the events based on the set inclusion.
We introduce them and their categorical representation modeling operations from process (\autoref{sec:st_fam}).

One can also obtain a causal semantics of a process calculus, by decorating its LTS.
In \autoref{sec:causal_ccs} we briefly show how to interpret CCS terms in configuration structures and how to decorate its LTS to derive causal information from the transitions.

Lastly, we introduce configuration structure for a restricted class of RCCS process, called \emph{singly labelled} (\autoref{def:singlylabelled}).
They are essentially an address in the configuration structure of the underlying, \enquote{original} CCS term.
We then introduce the LTS of those configuration structures and prove their operational correspondence with the reversible syntax (\autoref{lem:bisim_stfam_ccs}).

\subsection{Configuration structures as a causal model}
\label{sec:st_fam}

Let \(E\) be a set, \(\subseteq\) be the usual set inclusion and \(C\) be a family of subsets of \(E\).
For \(X\subseteq C\), \(X\) is \emph{compatible}, denoted \(X\compa\), if \(\exists y\in C\) finite such that \(\forall x\in X\), \(x\subseteq y\).

\begin{definition} [Configuration structures]
	\label{def:conf_str}
	A \emph{configuration structure} \(\conf\) is a triple \((E,C,\subseteq)\) where \(E\) is a set of events, \(\subseteq\) is the set inclusion and \(C\subseteq\power(E)\) is a set of subsets satisfying:
	\begin{itemize}
		\item \emph{finiteness}:
		      \(\forall x \in C, \forall e \in x, \exists z \in C\) finite such that \( e\in z\) and \(z \subseteq x\),
		\item \emph{coincidence freeness}:
		      \(\forall x \in C, \forall e, e' \in x, e\neq e'\Rightarrow \big(\exists z \in C, z \subseteq x\text{ and }(e\in z \iff e'\notin z)\big)\),
		\item \emph{finite completness}:
		      \( \forall X \subseteq C \text{ if } X \compa \text{ then } \bigcup X \in C\),
		\item \emph{stability}:
		      \(\forall x,y\in C \text{ if }x\cup y\in C\text{ then }x\cap y\in C\).
	\end{itemize}
	We denote \(\confzero\) the configuration structure with \(E = \emptyset\).
\end{definition}

Intuitively, events are the actions occurring during the run of a process, while configurations represents computational states.
The first axiom, \emph{finiteness}, guarantees that for each event the set of causes is finite. \emph{Coincidence freeness} states that each computation step consists of a single event.
Axioms \emph{finite completness} and \emph{stability} are more abstract and are better explained on some examples.
Consider the structures in \autoref{fig:counterex_stfam}: the structure \ref{fig:counterex_stfam_a} does not satisfy the second axiom, as two events occur in a single step.
The structure \ref{fig:counterex_stfam_b} does not satisfy finite completeness.
Intuitively, configuration structures cannot capture \enquote{pairwise} concurrence.
Finally, the structure \ref{fig:counterex_stfam_c} does not satisfies stability and the intuition is that the causes of event \(e_3\) are \emph{either} \(e_1\) \emph{or} \(e_2\), but not both.

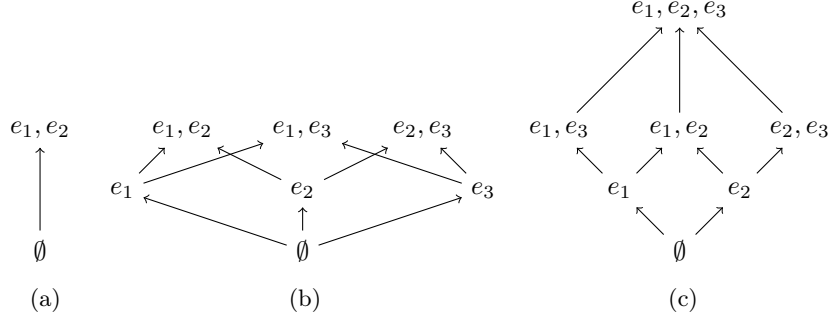
\begin{figure}
	{\centering
		\begin{minipage}[b]{.1\linewidth}
			\begin{tikzpicture}[scale=0.8]
				\node (emptyset) at (0,-1) {\(\emptyset\)};
				\node (ab) at (0,1) {\(e_1,e_2\)};
				\draw [->] (emptyset) -- (ab);
			\end{tikzpicture}
			\subcaption{}\label{fig:counterex_stfam_a}
		\end{minipage}
		\begin{minipage}[b]{.45\linewidth}
			\begin{tikzpicture}[scale=0.8]
				\node (emptyset) at (0,-1) {\(\emptyset\)};
				\node (a) at (-3,0) {\(e_1\)};
				\node (b) at (0,0) {\(e_2\)};
				\node (c) at (3,0) {\(e_3\)};
				\node (ab) at (-2, 1) {\(e_1,e_2\)};
				\node (ac) at (0,1) {\(e_1,e_3\)};
				\node (bc) at (2,1) {\(e_2,e_3\)};
				\draw [->] (emptyset) -- (a);
				\draw [->] (emptyset) -- (b);
				\draw [->] (emptyset) -- (c);
				\draw [->] (a) -- (ab);
				\draw [->] (a) -- (ac);
				\draw [->] (b) -- (bc);
				\draw [->] (b) -- (ab);
				\draw [->] (c) -- (bc);
				\draw [->] (c) -- (ac);
			\end{tikzpicture}
			\subcaption{}\label{fig:counterex_stfam_b}
		\end{minipage}
		\begin{minipage}[b]{.36\linewidth}
			\begin{tikzpicture}[scale=0.8]
				\node (emptyset) at (0,-1) {\(\emptyset\)};
				\node (b) at (-1,0) {\(e_1\)};
				\node (c) at (1,0) {\(e_2\)};
				\node (bcp) at (0, 1) {\(e_1,e_2\)};
				\node (ba) at (-2,1) {\(e_1,e_3\)};
				\node (ca) at (2,1) {\(e_2,e_3\)};
				\node[align=center] (bapc) at (0,3) {\(e_1,e_2,e_3\)};
				\draw [->] (emptyset) -- (b);
				\draw [->] (emptyset) -- (c);
				\draw [->] (b) -- (bcp);
				\draw [->] (c) -- (bcp);
				\draw [->] (b) -- (ba);
				\draw [->] (c) -- (ca);
				\draw [->] (ba) -- (bapc);
				\draw [->] (ca) -- (bapc);
				\draw [->] (bcp) -- (bapc);
			\end{tikzpicture}
			\subcaption{}\label{fig:counterex_stfam_c}
		\end{minipage}
		\caption{Structures that are not coincidence free, finite complete and stable, respectively}
		\label{fig:counterex_stfam}
	}
\end{figure}

\begin{notations}
	In a configuration \(\conf{C}\), if \(x, x' \in C\), \(e\in E \), \(e\notin x \) and \(x' = x \cup \{e\} \), then we write \(x\redl{e}x'\).
\end{notations}

\begin{definition}[Labelled configuration structure]
	A \emph{labelled configuration structure} \(\conf=(E,C,\labl)\) is a configuration structure endowed with a \emph{labelling function} from events to labels \(\labl:E\to\labels\).
\end{definition}
From now on, we will only consider configurations structures that are labelled, so we omit that adjective in the following.

Now we define morphisms on configurations structures that permits to form a category whose objects are configuration structures.
Intuitively, morphisms model the inclusion or refinement relations between processes.
Process algebras' operators are then extended to configuration structures, which makes it a modular model.

\begin{definition}[Category of configuration structures]
	\label{def:st_fam_morph}
	A morphism of configurations structures \(f:(E_1,C_1,\labl_1)\to(E_2,C_2,\labl_2)\) is a partial function on the underlying sets \(f:E_1\rightharpoonup E_2\) that is:
	\begin{itemize}
		\item \emph{configurations preserving}: \(\forall x\in C_1, f(x)=\{f(e) \setst e\in x\}\in C_2\),
		\item \emph{local injective}: \(\forall x\in C_1, \forall e_1, e_2\in x, f(e_1)=f(e_2) \implies e_1=e_2\),
		\item \emph{label preserving}: \(\forall x \in C_1, \forall e \in x, \labl_1(e)=\labl_2(f(e))\).
	\end{itemize}
	An isomorphism on configuration structures is denoted \(\iso\).
\end{definition}

\begin{definition}[Operations on configuration structures]
	\label{cat-op-def}
	Let \(\conf_1=(E_1,C_1,\labl_1)\), \(\conf_2=(E_2,C_2,\labl_2)\) be two configuration structures and set \(E^\star=E \cup \{\star\}\).
	\begin{description}
		\item[Product]
		Let \(\star\) denote \emph{undefined} for a partial function.
		Define \emph{the product of \(\conf_1\) and \(\conf_2\)} as \(\conf=\conf_1\times\conf_2\), for \(\conf=(E,C,\labl)\), where \(E=E_1\times_{\star} E_2\) is the product in the category of sets and partial functions%
		\footnote{The category of sets and partial functions has sets as objects and functions that can take the value \(\star\) as morphisms~\cite[Appendix A]{Winskel1982}.}:
		\[
			E_1\times_{\star} E_2 =
			\begin{multlined}[t]
				\{(e_1,\star)\mid e_1\in E_1\}\cup\{(\star,e_2)\mid e_2\in E_2\}\\
				\cup\{(e_1,e_2)\mid e_1\in E_1, e_2\in E_2\}
			\end{multlined}
		\]
		with the projections \(p_1:E\to E_1\cup\{\star\}\), \(p_2:E\to E_2\cup\{\star\}\).
		Define the projections \(\pi_1:(E,C)\to (E_1,C_1)\), \(\pi_2:(E,C)\to (E_2,C_2)\) such that the following holds, for \(e\in E\) and \(x\in C\):
		\begin{itemize}
			\item \(\pi_1(e)=p_1(e)\) and \(\pi_2(e)=p_2(e)\);
			\item \(\pi_1(x)\in C_1\) and \(\pi_2(x)\in C_2\);
			\item \(\forall e,e'\in x\), if \(\pi_1(e)=\pi_1(e')\neq\star\) or \(\pi_2(e)=\pi_2(e')\neq\star\) then \(e=e'\);
			\item \(\forall e \in x, \exists z\subseteq x\) finite s.t. \(\pi_1(x) \in C_1\), \(\pi_2(x)\in C_2\) and \(e\in z\);
			\item \(\forall e, e' \in x, e\neq e'\Rightarrow \exists z\subseteq x\) s.t. \(\pi_1(z) \in C_1\), \(\pi_2(z)\in C_2\) and \((e\in z \iff e'\notin z)\).
		\end{itemize}
		The labelling function \(\labl\) is defined as follows:
		\[
			\labl(e) = \begin{cases}
			\labl_1(e_1) & \text{ if }\pi_1(e)=e_1,\pi_2(e)=\star \\
			\labl_2(e_2) & \text{ if }\pi_1(e)=\star,\pi_2(e)=e_2 \\
			(\labl_1(e_1),\labl_2(e_2)) & \text{ otherwise}
			\end{cases}
		\]

		\item[Coproduct]
		Define \emph{the coproduct of \(\conf_1\) and \(\conf_2\)} as \(\conf=\conf_1+\conf_2\), for \(\conf=(E,C,\labl)\), where \(E=(\{1\}\times E_1)\cup(\{2\}\times E_2)\) and \(C=\{\{1\}\times x \setst x\in C_1\}\cup\{\{2\}\times x \setst x\in C_2\}\).
		The labelling function \(\labl\) is defined as \(\labl(e)=\labl_i(e_i)\) when \(e_i\in E_i\) and \(\pi_i(e_i)=e\).

		\item[Restriction]
		Let \(E'\subseteq E\) and define \emph{the restriction of a set of events} as \((E,C,\labl)\restr E'=(E',C',\labl')\) where \(x\in C'\iff x\in C\) and \(x\subseteq E'\). %

		\emph{The restriction of a name} is then \((E,C,\labl)\restr a \coloneqq (E,C,\labl)\restr E_a\) where \(E_a=\{e\in E \setst \labl(e) \in \{a, \out{a}\}\}\).
		For \(a_1, \hdots, a_n\) a list of names, we define similarly \(\restr \cup_{1 \leqslant i \leqslant n}E_{a_i}\).

		\item[Prefix]Let \(\lambda\) be the label of an event and define \emph{the prefix operation on configuration structures} as \(\alpha.(E,C,\labl)=(e\cup E,C',\labl')\), for \(e\notin E\) where \(x'\in C' \iff \exists x\in C\), \(x'=x\cup e\) and \(\labl'(e) = \alpha\), and \(\forall e' \neq e\), \(\labl'(e') = \labl(e')\).

		\item[Relabelling]
		Define \emph{the relabelling of a configuration structure} as \(\conf_1\circ\labl=(E_1,C_1,\labl)\), where \(\conf_1=(E_1,C_1,\labl_1)\) and \(\labl:E_1\to\labels\) is a labelling function.

		\item[Parallel composition]
		Define parallel composition \(\conf = \big((\conf_1\times\conf_2)\circ\labl\big)\restr E\) as the application of product, relabelling and restriction, with \(\labl\) and \(E\) defined below.
		\begin{itemize}
			\item First, \(\conf_1\times\conf_2 =\conf_3\) is the product with \(\conf_3=(E_3,C_3,\labl_3)\);
			\item Then, \(\conf'=\conf_3\circ\labl\) with \(\labl\) defined as follows:
			      \[
			      	\labl(e) = \begin{cases}
			      	\labl_3(e) & \text{ if }\labl_3(e)\in\{a;\out{a}\} \\
			      	\tau & \text{ if }\labl_3(e)\in\{(a,\out{a}); (\out{a},a)\}\\
			      	0 & \text{ otherwise }
			      	\end{cases}
			      \]
			\item Finally, \(\conf= (E_1 \times_{\star} E_2,C_3,\labl) \restr E\) is the resulted configuration structure, where \(E=\{e\in E_3 \setst \labl(e)\neq 0\}\).
		\end{itemize}
	\end{description}
\end{definition}

In the definition of the product, the conditions guarantee that \(\conf_1\times\conf_2\) is the product in the category of configuration structures and that the projections \(\pi_1\), \(\pi_2\) are morphisms.
In particular, the third condition ensures that the projections are local injective, the fourth and fifth enforce finiteness and coincidence-freeness axioms in the resulted configuration structure.

\begin{definition}[Causality]
	\label{def:causality}
	Let \(x\in C\) and \(e_1,e_2\in x\) for \((E,C,\labl)\) a configuration structure.
	Then we say that \emph{\(e_1\) happens before \(e_2\) in \(x\)} or that \emph{\(e_1\) causes \(e_2\) in \(x\)}, written \(e_1\leqslant_x e_2\), iff \(\forall x_2\in C, x_2\subseteq x, e_2\in x_2\implies e_1\in x_2\).
\end{definition}

Configurations can also be interpreted as \emph{temporal} observations~\cite[Chap.~5]{Cristescu2015}, instead of causal orders present in the structure of a term.
Refering to the order as \emph{happens before} instead of causality highlights the observational nature of the order.

Morphisms on configuration structures reflect causality.
Let \(f:\conf_1\to\conf_2\) a morphism and \(x\in C_1\) a configuration.
Then
\[ \forall e_1,e_2\in x,\text{ if } f(e_1) \leqslant_{f(x)} f(e_2)\text{ then }e_1\leqslant_x e_2.\]

However, morphisms do not preserve causality in general.
In the case of a product we can show that all \emph{immediate} causalities are due to one of the two configurations structures.
Stated differenlty, a context can add but cannot remove causality in the process.

\begin{definition}[Immediate causality]
	Let \(e,e'\) be two events in a configuration \(x\) for a configuration structure \((E,C,\labl)\).
	Denote \(e\to_x e'\) if \(e\) is an \emph{immediate cause} for \(e'\) in \(x\), that is \(e<_xe'\) and \(\nexists e''\) such that \(e<_x e''<_x e'\).
\end{definition}
Note that we overload the notation \(e\to_x e'\) however as it is defined on events only, it is not ambiguous.

\begin{proposition}
	\label{prop:cause_projection}
	Let \(x\in\conf_1\times\conf_2\).
	Then \(e_1\to_x e_2\iff\) either \(\pi_1(e_1)<_{\pi_1(x)} \pi_1(e_2)\) or \(\pi_2(e_1)<_{\pi_2(x)} \pi_2(e_2)\).
\end{proposition}

\begin{proof}
	The proof~\cite[Proposition 6]{Cristescu2015} follows by contradiction, using that if \(x\) is a configuration in \(\conf\) and if \(e\in x\) is such that \(\forall e'\in x\), \(e\not <_x e'\), then \(x\setminus e\) is a configuration in \(\conf\).
\end{proof}

\subsection{Operational semantics, correspondence and equivalences}
\label{sec:causal_ccs}
Configuration structures are a causal model for CCS~\cite{Winskel1995} in which the computational states are the configurations and the forward executions are dictated by set inclusion.
To show the correspondence with CCS (\autoref{lem:bisim_stfam_ccs}), one defines an operational semantics on configurations structures (\autoref{def:transition_stable}) that erases the part of the structure that is not needed in future computations.
In order to define a reversible semantics on configurations structures a second LTS is introduced (\autoref{def:reverLTS}), that instead of being defined on configurations structures is defined on the configurations of a stable family.
Thus forward and backward moves are simply the set inclusion relation and its opposite, respectively.

The soundness of the model is proved by defining an operational semantics on configurations structures and showing an operational correspondence between the two worlds.

\begin{definition}[Encoding a CCS term]
	\label{def:encoding-CCS-cf}
	Given \(P\) a CCS term, its encoding \(\enc{P}\) as a configuration structure is built inductively, using the operations of \autoref{cat-op-def}:
	\begin{align*}
		\enc{a.P}     & =a.\enc{P}          &   &   &
		\enc{\out{a}.P}&=\out{a}.\enc{P}\\
		\enc{P\mid Q} & =\enc{P}\mid\enc{Q} &   &   &
		\enc{P+Q}&=\enc{P}+\enc{Q}\\
		\enc{P\bs a}  & =\enc{P}\restr E_a  &   &   &
		\enc{0}&=\confzero
	\end{align*}
\end{definition}

\begin{definition}[LTS on configurations structures]
	\label{def:transition_stable}
	Let $\conf=(E,C,\labl)$ be a configuration structure.
	Define $\conf\setminus{e} = (E\setminus{e},C',\labl')$ where $\labl'$ is the restriction of $\labl$ to the set $E\setminus{e}$ and
	\[ x\in C'\iff x\cup\{e\}\in C.\]
	We easily make sure that $\conf\setminus{e}$ is also a configurations structures.

	We define a LTS on configurations structures thanks to the relation $\conf\redl{e}\conf\setminus{e}$, and we extend the notation to $\conf\redl{\labl(e)}\conf\setminus{e}$ and to $\conf\redl{x}\conf\setminus{x}$, for $x$ a configuration.
\end{definition}

\begin{lemma}[Operational correspondence between a process $P$ and its encoding $\enc{P}$]
	\label{lem:bisim_stfam_ccs}
	Let $P$ a process and $\enc{P}=(E,C,\labl)$ its interpretation.
	\begin{enumerate}
		\item $\forall \alpha$, $P'$ such that $P\redl{\alpha}P'$, $\exists e\in E$ such that $\labl(e)=\alpha$ and $\enc{P}\setminus{e}\iso\enc{P'}$;
		\item $\forall e\in E$, if $\{e\}\in C$ then $\exists P'$ such that $P\redl{\labl(e)}P'$ and $\enc{P}\setminus{e}\iso\enc{P'}$.
	\end{enumerate}
\end{lemma}
The above lemma shows that \emph{labelled} transitions are in correspondence, but labels are just a tool for compositionality.
The main result is that a process and its encoding simulate each others \emph{reductions}.

\begin{theorem}[Operational correspondence with CCS]
	Let $P$ a process and $\enc{P}=(E,C,\labl)$ its encoding.
	\begin{enumerate}
		\item $\forall P'$ such that $P\redl{\tau}P'$, $\exists \{e\}\in C$ closed such that $\enc{P}\setminus{e}\iso\enc{P'}$;
		\item $\forall e\in E$, $\{e\}\in C$ closed, $\exists P'$ such that $P\redl{\tau}P'$ and $\enc{P}\setminus{e}\iso\enc{P'}$.
	\end{enumerate}
\end{theorem}

Multiple equivalence relations on configuration structures have been defined and studied~\cite{Glabeek1989, Nicola1990,Phillips2012,Baldan2014,Vogler1993}.
Among them, hereditary history preserving bisimulation (HHPB)~\cite{Bednarczyk1991} equates structures that simulate each others' forward and backward moves and thus connects configuration structures to reversibility.
It is considered a canonical equivalence on configuration structures as it respects the causality and concurrency relations between events and admits a categorical representation~\cite{Joyal1996b}.

Those connections between reversibility and causal models sheds a new light on what help are the meaningful equivalences on reversible processes.
Consequently, one apply them in the operational setting.
We begin by modifying the definition of our LTS on configuration structures to include backward moves as well.

\begin{definition}[A reversible LTS on configuration structures]
	\label{def:reverLTS}
	Consider \((E, C,\labl)\) a configuration structure.
	For \(x\in C\), \(e\in E\) define \(x\redl{e}x'\) iff \(x'=x\cup\{e\}\) and \(x\redl{\alpha}x'\) if additionally, \(\labl(e)=\alpha\).
	The backward moves are defined as \(x\revredl{e}x'\) and \(x\revredl{\alpha}x'\) if \(x=x'\cup\{e\}\) and \(\labl(e)=\alpha\).

	Denote \(x\fbl{e}x'\) when either \(x\redl{e}x'\) or \(x\revredl{e}x'\).
\end{definition}

Such a LTS naturally satisfies elementary criterion that one wold expect from a LTS~\cite[Chap.~2]{Cristescu2015}.

\begin{definition}[HHPB {\cite[Definition 1.4]{Bednarczyk1991}}]
	\label{def:hhpb}
	A \emph{hereditary history preserving bisimulation} on labelled configuration structures is a symmetric relation \(\rel\subseteq C_1\times C_2\times\power(E_1\times E_2)\) such that \((\emptyset,\emptyset,\emptyset)\in {\rel}\) and if \((x_1,x_2,f)\in {\rel}\), then
	\begin{align*}
		f\text{ is a label and order preserving bijection between }x_1\text{ and }x_2\notag                                                                \\
		x_1\redl{e_1}x_1'\implies \exists x_2'\in C_2 \text{ \st~} x_2\redl{e_2}x_2' \text{ and } f=f'\restr x_1, (x_1',x_2',f')\in {\rel} \notag          \\
		x_1 \revredl{e_1} x_1'\implies \exists x_2'\in C_2 \text{ \st~} x_2 \revredl{e_2} x_2'\text{ and } f'=f\restr x_2, (x_1',x_2',f')\in {\rel} \notag
	\end{align*}
\end{definition}

It is known~\cite{Phillips2007} that hereditary history preserving bisimulation corresponds to the back-and-forth bisimulation (\autoref{def:sbfc_rccs}), in the following sense: CCSK~\cite{Phillips2007b} is proven to satisfy the \enquote{the axioms of reversibility}~\cite[Definition~4.2]{Glabbek1996}), so that its LTS is \emph{prime}.
Then, this LTS is represented as a process graph, on which the forward-reverse bisimulation~\cite[Definition 5.1]{Phillips2007}---our back-and-forth bisimulation (\autoref{def:sbfc_rccs})--- is defined.
Finally, configuration graphs and hereditary history-preserving bisimulation are defined from configuration structures, and both relation are proven to coincide.~\cite[Theorem~5.4, p.~105]{Phillips2007}.

\subsection{Configuration structures for RCCS }
\label{sec:stfam_rccs}
All the possible future behaviours of a process without memory are present in its encoding as a configuration structure.
All alike, we want our encoding of processes with memory to record both their past \emph{and} their future, so that they can evolve in both directions, as process do.
To this end, we encode RCCS terms as a configuration in the configuration structure of their origins (\autoref{def:encod_rccs}).
Then, we show an operational correspondence between RCCS terms and their encoding.

To determine which configuration corresponds to the computational state of the term we are encoding, we need to uniquely identify a process from its past and future. %
However, as the following example illustrates, this is not always possible: \autoref{rem-concur} explains the limitations of the encoding we are going to develop.

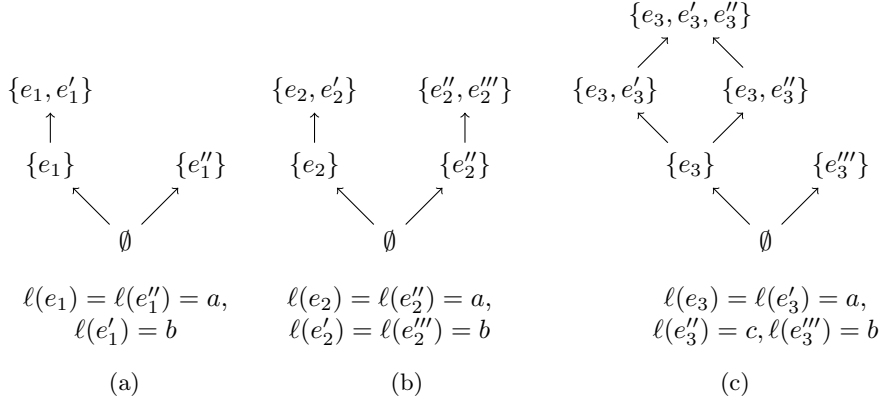
\begin{figure}
	\begin{minipage}[t]{.28\linewidth}
		\begin{tikzpicture}
			\node (emptyset) at (0, -1) {\(\emptyset\)};
			\node (a) at (-1, 0) {\(\{e_1\}\)};
			\node (b) at (1, 0) {\(\{e_1''\}\)};
			\node (ab) at (-1, 1) {\(\{e_1, e'_1\}\)};
			\draw [->] (emptyset) -- (a);
			\draw [->] (emptyset) -- (b);
			\draw [->] (a) -- (ab);
			\node[align=center] (labels) at (0, -2) {\(\labl(e_1) = \labl(e_1'') = a\),\\ \(\labl(e'_1) = b\)};
		\end{tikzpicture}
		\subcaption{}\label{ex_unif_a}
	\end{minipage}
	\begin{minipage}[t]{.32\linewidth}
		\begin{tikzpicture}
			\node (emptyset) at (0, -1) {\(\emptyset\)};
			\node (a) at (-1, 0) {\(\{e_2\}\)};
			\node (b) at (1, 0) {\(\{e''_2\}\)};
			\node (ab) at (-1, 1) {\(\{e_2, e'_2\}\)};
			\node (ba) at (1, 1) {\(\{e''_2, e'''_2\}\)};
			\draw [->] (emptyset) -- (a);
			\draw [->] (emptyset) -- (b);
			\draw [->] (a) -- (ab);
			\draw [->] (b) -- (ba);
			\node[align=center] (labels) at (0, -2) {\(\labl(e_2) = \labl(e_2'') =a\),\\ \(\labl(e_2') = \labl(e_2''') =b\)};
		\end{tikzpicture}
		\subcaption{}\label{ex_unif_b}
	\end{minipage}
	\begin{minipage}[t]{.38\linewidth}
		\begin{tikzpicture}
			\node (emptyset) at (0, -1) {\(\emptyset\)};
			\node (a) at (-1, 0) {\(\{e_3\}\)};
			\node (b) at (1, 0) {\(\{e_3'''\}\)};
			\node (aa1) at (-2, 1) {\(\{e_3, e_3'\}\)};
			\node (aa2) at (0, 1) {\(\{e_3, e_3''\}\)};
			\node (aaa) at (-1, 2) {\(\{e_3, e_3',e_3''\}\)};
			\draw [->] (emptyset) -- (a);
			\draw [->] (emptyset) -- (b);
			\draw [->] (a) -- (aa1);
			\draw [->] (a) -- (aa2);
			\draw [->] (aa1) -- (aaa);
			\draw [->] (aa2) -- (aaa);
			\node[align=center] (labels) at (0, -2) {\(\labl(e_3) = \labl(e_3') =a, \)\\ \(\labl(e_3'') = c, \labl(e_3''') = b\)};
		\end{tikzpicture}
		\subcaption{}\label{ex_unif_c}
	\end{minipage}
	\caption{Encoding RCCS in configurations structures}
	\label{ex_unif}
\end{figure}

\begin{example}
	\label{ex:auto}
	\begin{enumerate}
		\item The process \(P=a.b+a\) is interpreted as the configuration structure in \autoref{ex_unif_a}.
		      Let us consider the execution \(\emptymem\rhd P\redl{a}R\).
		      To determine which of the configurations labelled \(a\) correspond to \(R\) we have to consider the future of \(R\) as well.

		\item Hence we choose a configuration that respects the past and the future of \(R\), but such a configuration might not be unique.
		      Let \(Q=a.b + a.b\) be a process whose configuration structure is in \autoref{ex_unif_b}.
		      For the trace \(P \redl{a} b\) there is no way to choose between the two configurations labelled \(a\).
	\end{enumerate}
\end{example}

The situation of \autoref{ex:auto} is generalizable to any process whose reduction may lead to a process of the form \(a.P+a.P\) or \(a.P\mid a.P\).
We consider then a restricted class of processes, as discussed in the following remark.
\begin{remark}[On auto-concurrency and others limitations]
	\label{rem-concur}
	In the following, we need to uniquely identify configurations based solely on the labels and orders of the \emph{open} (\ie non-synchronized) events.
	As seen in \autoref{ex:auto}, this is not possible in encoding of processes that may reduce to the form \(a.P \mid a.Q\) or \(a.P+a.Q\).

	The first kind of process is characterised by an \emph{auto-concurrency} condition~\cite[Definition 9.5]{Glabbeek2001}.
	We need a stronger condition, a sort of \emph{auto-conflict}, to forbid the second type of process.
\end{remark}

\begin{definition}[Singly labelled configuration structures and processes]
	\label{def:singlylabelled}
	A configuration structure \(\conf\) is \emph{singly labelled}, or \emph{without auto-concurrency nor auto-conflict} if \(\forall x\in\conf\) and \(\forall e,e'\notin x\) we have that
	\[\big( x\redl{e} y, x\redl{e'}y'\text{ and }\labl(e)=\labl(e')\big)\implies e = e'.\]

	A process is singly labelled if its encoding as a configuration structure is.
\end{definition}

Remark that being singly labelled does not mean that each label has to occurs only once in a process: whereas \(a \mid b.a \) is not, since after firing \(b\) two transitions labelled \(a\) can be fired, \(a.a\) and \(a.b + b\) are singly labelled.
However, a syntactical definition of this restriction cannot be inductively defined, since \(P\) and \(Q\) might be singly labelled, but not \(P \mid Q\) nor \(P + Q \).

The following encoding, and all the results that use it, require the process to be singly labelled (on top of being coherent, if they are reversible).
This restriction could probably be removed at the price of a \emph{tagging} of the occurrences of names, maybe in the spirit of the \emph{localities}~\cite{Boudol1998b}.

\begin{definition}[Encoding RCCS processes in configurations structures]
	\label{def:encod_rccs}
	Let \(R\) be a singly labelled process and \(\conf = \encc{\erase{\orig{R}}}\) the encoding (\autoref{def:encoding-CCS-cf}) of its \enquote{memory-less} origin (\autoref{def:rccs_ccs}).

	We first need the function \(\mathaddress_{\conf}\), defined as:
	\begin{align}
		\funaddress{\conf}{x}{f,R_1\redl{i:\alpha}R_2\red^{\star} R_3} & =\funaddress{\conf}{x\cup\{e\}}{f\cup\{e\leftrightarrow i\},R_2\red^{\star} R_3} \label{fun1} \\
		\funaddress{\conf}{x}{f,R_2\red^{\star}R_3}                    & = x \text{ if } R_2 =R_3 \label{fun2}
	\end{align}
	Where in \eqref{fun1} \(e\) is such that
	\begin{itemize}
		\item \(\labl(e)=\alpha\);
		\item \(x\cup\{e\}\in \conf\);
		\item \(j<_{R_2}i\iff f(j)<_{x\cup\{e\}}e\);
		\item and \(\enc{\erase{R_2}}= \big(\conf\setminus (x\cup\{e\})\big)\).
	\end{itemize}
	Now we define the encoding of \(R\) in configuration structure by induction on the trace (\autoref{def:trace}) \(\sigma: \orig{R}\red^{\star} R\), as \(\encr{R}_{\sigma} = (\conf,\funaddress{\conf}{\emptyset}{\emptyset,\sigma})\).

\end{definition}
We show in \autoref{prop-soundness-rccs} that the function is well defined, \ie for every singly labelled process \(R\) and for every trace \(\sigma: \orig{R}\red^{\star}R\) there exists a unique configuration in \(\enc{\erase{\orig{R}}}_{\sigma}\) defined as above.

\begin{example}
	A first simple example is the encoding of a process with an empty memory.
	Let \(S = \emptymem \rhd P\), \(\erase{\orig{S}} = P\) and \(\enc{S}_{\tempty}=(\enc{P}, \emptyset)\).

	Let us show how to compute the encoding of the process
	\[R=\mem{2,a,0}.\fork.\mem{1,a,b}\rhd 0\mid \fork.\mem{1,a,b}\rhd a.\]
	We backtrack to its origin and obtain \(\orig{R}=\emptymem \rhd a.(a \mid c)+b\).
	The term is encoded in the configuration structure in \autoref{ex_unif_c}.
	We apply the function \(\funaddress{\conf}{\emptyset}{\orig{R}\red^{\star} R}\) on the trace
	\begin{align*}
		\emptymem\rhd a.(a \mid c)+b & \fwlts{1}{a} \mem{1,a,b}\rhd (a\mid c)                                                     \\
		                             & \qquad \congru (\fork.\mem{1,a,b}\rhd a)\mid (\fork.\mem{1,a,b}\rhd c )                    \\
		                             & \qquad \qquad \fwlts{2}{a} \mem{2,a,0}.\fork.\mem{1,a,b}\rhd 0\mid \fork.\mem{1,a,b}\rhd c \\
		                             & \qquad \qquad \qquad = R'
	\end{align*}

	The configuration corresponding to \(R\) is then \(\{e_3,e_3'\}\).
\end{example}

Let us show that the encoding is correct, and in particular that the function \(\mathaddress_{\conf}\) is well defined.

\begin{proposition}[Soundness of the RCCS encoding]
	\label{prop-soundness-rccs}
	Let \(P\) be a singly labelled process and \(\conf=\encc{P}\) its encoding.
	Then for any \(R\) reachable from \(\emptymem\rhd P\)
	there exists a unique \(x\in\conf\) such that \(\funaddress{\conf}{\emptyset}{\emptyset,\emptymem\rhd P\red^{\star}R}=x\).
\end{proposition}

\begin{proof}
	From \autoref{lem:rearrange_trace}, we consider the trace \(\orig{R}\red^{\star} R\) to be only forward.
	We proceed by induction on the trace \(\orig{R}\red^{\star} R\).
	For the inductive case we have the trace \(\orig{R}\red^{\star} R_n\) %
	and
	\(\funaddress{\conf}{\emptyset}{f_n,\orig{R}\red^{\star} R_n}=x_n\), for \(x_n\in\conf\), \(f_n\) a label and order preserving bijection between \(x_n\) and \(R_n\), and such that \(\encc{\erase{R_n}}=\conf\setminus x_n\).
	We have to show that for the trace \(\orig{R}\red^{\star} R_n\redl{i:a} R_{n+1}\) %
	there exists a unique configuration \(x_{n+1}\in \conf\) such that
	\begin{align*}
		\funaddress{\conf}{\emptyset}{\emptyset,\orig{R}\red^{\star} R_n \redl{i:\alpha}R_{n+1}} =x_{n+1}
		\shortintertext{ and }
		\enc{\erase{R_{n+1}}}=\conf\setminus {x_{n+1}}.
	\end{align*}
	We have that
	\[
		\funaddress{\conf}{\emptyset}{\emptyset,\orig{R}\red^{\star} R_n\redl{i:\alpha}R_{n+1}}=\funaddress{\conf}{x_n}{f_n,R_n\redl{i:\alpha}R_{n+1}}
	\]
	Hence we have that \(x_{n+1}=x\cup\{e\}\), \(f_{n+1}=f_n\cup\{e\leftrightarrow i\}\) and we have to show that there exists a unique \(e\in\conf\) such that \(\labl(e)=\alpha\) and
	\[\encc{\erase{R_{n+1}}}=\conf\setminus (x_n\cup\{e\}).\]
	However, if such an \(e\) exists then \(e\in\encc{\erase{R_n}}\) and
	\[\conf\setminus (x_n\cup\{e\})=\encc{\erase{R_n}}\setminus{e}.\]
	Hence we reason on the transition \(R_n\redl{i:\alpha}R_{n+1}\) to show that there exists a unique \(\{e\}\in\encc{\erase{R_n}}\) such that \(\encc{\erase{R_{n+1}}}=\encc{\erase{R_n}}\setminus{e}\).
	We consider only the case \(\alpha=a\), the rest being similar.
	Using structural congruence it is possible to rewrite \(R_n\) and \(R_{n+1}\) as follows
	\[
		R_n\congru(m_1\rhd a.P_1 \mid R_2)\bs (b_1\dots b_n)\qquad R_{n+1}\congru(m_1\rhd P_1 \mid R_2)\bs (b_1\dots b_n)
	\]
	and hence, for \(\erase{R_2}=P_2\),
	\[\erase{R_n}=(a.P_1 \mid P_2)\bs (b_1\dots b_n)\qquad\erase{R_{n+1}}=(P_1 \mid P_2)\bs (b_1\dots b_n).\]
	We have then to show that
	\[ \encc{(P_1 \mid P_2)\bs (b_1\dots b_n)}=\encc{(a.P_1 \mid P_2)\bs (b_1\dots b_n)}\setminus{e}.\]
	From \autoref{lem:bisim_stfam_ccs} such an event exists.
	To show its uniqueness, consider \(R_n\congru m_1\rhd a.P_1 \mid \big( m_2\rhd a.P_2 \mid R_2\big)\).
	Either \(m_1=m_2\) in which case the process exhibits auto-concurrency, or \(m_1\neq m_2\) and in this case the condition \(j<_{R_{n+1}}i\iff f_n(j)<_{x_n\cup\{e\}}e\) from the definition of the encoding, points to either \(m_1\) or \(m_2\).

	Let us prove that \(\forall x\in\encc{(P_1 \mid P_2)\bs (b_1\dots b_n)}\), \(x\in\encc{(a.P_1 \mid P_2)\bs (b_1\dots b_n)}\setminus{e}\).
	The other direction is similar.
	Let us unfold the encoding of \autoref{def:encoding-CCS-cf} using the operations on configurations structures of \autoref{cat-op-def}.
	\begin{align*}
		\encc{(P_1 \mid P_2)\bs (b_1\dots b_n)}   & =(\encc{P_1}\times\encc{P_2}) \restr \cup_{1 \leqslant i \leqslant n}E_{b_i}    \\
		\encc{(a.P_1 \mid P_2)\bs (b_1\dots b_n)} & = (\encc{a.P_1}\times\encc{P_2}) \restr \cup_{1 \leqslant i \leqslant n}E_{b_i}
	\end{align*}
	If \(x\in (\encc{P_1}\times\encc{P_2}) \restr \cup_{1 \leqslant i \leqslant n}E_{b_i}\) then
	\begin{equation}
		\label{eq:labl}
		\forall e\in x, \labl(e)\notin\{b_i,\out{b_i},0\}.
	\end{equation}
	Hence \(x\in(\encc{P_1}\times\encc{P_2})\).
	Let \(\pi_1\), \(\pi_2\) be the two projections defined by the product.
	Then
	\begin{equation}
		\label{eq:proj}
		\pi_1(x)\in\encc{P_1}\text{ and }\pi_2(x)\in\enc{P_2}.
	\end{equation}
	As \(\pi_1(x)\in\encc{P_1}\), and from the definition of \(\encc{a.P_1}\) we have that \(\exists e_1\), \(\labl(e_1)=a\) and such that \(\{e_1\}\cup \pi_1(x)\in a.\encc{P_1}\).
	From \autoref{eq:proj} we have that \(\exists x_2\in a.\encc{P_1}\times\encc{P_2}\) such that \(\pi_1(x_2)=\{e_1\}\cup \pi_1(x)\) and \(\pi_2(x_2)=\pi_2(x)\).
	Hence \(\exists !e\) such that \(\pi_1(e)=e_1\), \(\pi_2(e)=\star\) and \(x_2=\{e\}\cup x\).
	From \autoref{eq:labl} we have that \(x_2\in (a.\encc{P_1}\times\encc{P_2})\restr \cup_{1 \leqslant i \leqslant n}E_{b_i}\).
	We infer that if \(x\cup\{e\}\in(b_1\dots b_n)(a.\encc{P_1}\times\encc{P_2})\) then \(x\in\encc{(b_1\dots b_n)(a.P_1 \mid P_2)}\setminus{e}\).

	From \(\encc{\erase{R}}=\conf\setminus x_n\), we have that \(\forall y\in \encc{\erase{R}}\), \(y\cup x_n\in\conf\).
	In particular \(x_n\cup\{e\}\in\conf\).

	Let us denote \(x_{n+1}=x_n\cup\{e\}\).
	Remains to show that \(j<_{R_{n+1}}i\iff f(j)<_{x_{n+1}}e\).
	We show the implication \(j<_{R_{n+1}}i\implies f_n(j)<_{x_{n+1}}e\) and consider the immediate order for \(<_{R_{n+1}}\), as the order is transitive.
	From \(j<_{R_{n+1}}i\), we have that \(\mem{i,a}.\mem{j,b}\in R_{n+1}\), hence we retrieve a process \(R_k\) where \(b.a.P'\in R_k\).
	Hence the events \(f_n(j)\) and \(f_n(i)\) are causaly dependent in the configuration structure of \(\enc{\erase{R_k}}\), and therefore causally dependent in \(\conf\).
	For the other direction \(f_n (j)<_{x_{n+1}}e \implies j<_{R_{n+1}}i\) we show that \(\labl(f(j))\) and \(\labl(e)\) are causal in the origin process \(P\), hence they are causally dependent in the memory of \(R_{n+1}\).

	Hence \(\funaddress{\conf}{\emptyset}{\orig{R}\red^{\star} R_n\redl{a}R_{n+1}}=x_n\cup\{e\}\) with \(\encc{\erase{R_{n+1}}}=\encc{\orig{R_n}}\setminus (x_n\cup\{e\})\).
\end{proof}

\begin{remark}[On encoding RCCS]
	\label{rk:encode_rccs}
	Another encoding exists~\cite{Phillips2012}, but it is not compositional, since \(\encc{P_1 \mid P_2}\) is not defined as an operation on \(\encc{P_1}\) and \(\encc{P_2}\).
	Compositionality is important for the definition of contexts in configurations structures, in particular for the definition of congruence (\autoref{bisim-cs}).
\end{remark}

Let us now define a transition relation on configurations structures, useful in showing the operational correspondence between terms of RCCS and their encoding.

\begin{definition}[Reversible LTS in configurations structures]
	\label{def:lts_stfam}
	Define \((\encc{P}, x)\redl{\labl(e)}(\encc{P}, x\cup\{e\})\) for \(x\cup\{e\}\in\enc{P}\).
	Similarly to \autoref{def:rlts} we define \((\encc{P}, x)\revredl{\labl(e)}(\encc{P}, x\setminus{e})\), for some \(e\) such that \(x\setminus{e}\in\encc{P}\).
\end{definition}

We defined in \autoref{def:encod_rccs} the encoding of a process parametrically on a trace.
The following proposition shows that any trace from \(\orig{R}\) up to \(R\) leads to the same encoding.
\begin{proposition}
	\label{prop:unique_conf_trace}
	For all singly labelled processes \(R\) there exists \(x\) a configuration in \(\encc{\erase{\orig{R}}}\) such that \(\forall \sigma:\orig{R}\red^{\star} R \), \(\encr{R}_{\sigma}=x\) holds.
\end{proposition}

\begin{proof}
	Denote \(\conf=\encc{\erase{\orig{R}}} = (E,C,\labl)\).
	From the definition of \(\encr{R}_{\sigma}\) it suffices to show that for any configurations \(x,y\in\conf\) such that there exists a label and order preserving bijection between the two and such that \(\conf\setminus x=\conf\setminus y\), \(x=y\) holds.

	We prove it by induction on the size of \(x\) and \(y\).
	Suppose that there exists two events \(e,e'\) such that \(y=x\cup\{e\}\) and \(z=x\cup\{e'\}\) are configurations of \(\conf\) as well, with \(f:y\leftrightarrow z\).
	Since \(R\) is singly labelled (\autoref{def:singlylabelled}), if \(\labl(e)=\labl(e')\), then \(e=e'\).
\end{proof}

\begin{example}
	Consider the configuration structure in \autoref{ex_unif_c}, encoding the process \(P=a.(a \mid c)+b\).
	The process
	\[S= \big(\mem{2,a,0}.\fork.\mem{1,a,b}\rhd 0 \big) \mid \big(\mem{3,c,0}.\fork.\mem{1,a,b}\rhd 0\big)\]
	can be reached on the trace \(\sigma_1: \emptymem \rhd P\fwlts{1}{a}\fwlts{2}{a}\fwlts{3}{c} S\) or \(\sigma_1: \emptymem \rhd P\fwlts{1}{a}\fwlts{3}{c}\fwlts{2}{a} S\).
	However both traces lead to the same encoding of \(S\).
\end{example}

Hence we write \(\enc{R}\) instead of \(\enc{R}_{\sigma}\).
It is an essential property to prove the existence of a bisimulation relation between a process and its encoding.

\begin{lemma}[Operational correspondence between a \(R\) and \(\encr{R}\)]
	\label{lem:operational_corresp}
	Let \(R\) a process and \(\enc{R}=(\conf,x)\) its interpretation.
	\begin{enumerate}
		\item \(\forall \alpha\), \(S\) and \(i\in\ids\) such that \(R\fwlts{i}{\alpha}S\) then \(\encr{R}\redl{\alpha}\encr{S}\);
		\item \(\forall \alpha\), \(S\) and \(i\in\ids\) such that \(R\bwlts{i}{\alpha}S\) then \(\encr{R}\revredl{\alpha}\encr{S}\);
		\item \(\forall e\in E\), \((\conf,x)\redl{\labl(e)}(\conf, x\cup\{e\})\) then \(\exists S\), such that for some \(i\in\ids\), \(R\fwlts{i}{\alpha}S\) and \(\encr{S}=(\conf,x\cup\{e\})\).
		\item \(\forall e\in E\), \((\conf,x)\revredl{\labl(e)}(\conf, x\setminus{e})\) then \(\exists S\), such that for some \(i\in\ids\), \(R\bwlts{i}{\alpha}S\) and \(\encr{S}=(\conf,x\setminus{e})\).
	\end{enumerate}
\end{lemma}
\begin{proof}
	\begin{enumerate}
		\item As \(R\fwlts{i}{\alpha}S\), \(\orig{R}=\orig{S}\), we have that \(\encr{S}=(\conf,x_s)\), where
		      \(x_S= \funaddress{\conf}{\emptyset}{\orig{R}\red^\star S}=
		      \funaddress{\conf}{\emptyset}{\orig{R}\red^\star R\redl{\alpha} S}= x_R\cup\{e\}\) by \autoref{prop-soundness-rccs}.
		      As \(\encr{R}=(\conf,x_R)\) it follows that \((\conf,x_R)\redl{\alpha}(\conf,x_S)\).
		\item The proof for the backward direction is similar except that it uses the trace up to \(R\).
		      It uses \autoref{prop:unique_conf_trace}, which allows us to backtrack on any path from the emptyset and leading to \(x_R\).
		\item From \((\conf,x)\redl{\labl(e)}(\conf, x\cup\{e\})\) we have that \(x\cup\{e\}\in\conf\).
		      Then \(\{e\}\in\conf\setminus x\).
		      From \(\encr{R}=(\conf,x)\) we have that \(\conf\setminus x=\encc{\erase{R}}\), hence \(\{e\}\in\encc{\erase{R}}\).
		      We use \autoref{lem:bisim_stfam_ccs} and obtain that \(\exists P\) such that \(\erase{R}\redl{\labl(e)}P\).
		      Then due to the strong bisimulation between a RCCS term and its corresponding CCS term in \autoref{lem:corresp_ccs_rccs}, we have that, for some \(i\), \(R\fwlts{i}{\alpha}S\), where \(\erase{S}=P\).
		      That \(\encr{S}=(\conf,x\cup\{e\})\) follows from a similar argument to above and from \autoref{prop:unique_conf_trace}.\qedhere
		\item It is similar to the case above.
	\end{enumerate}
\end{proof}

\section{Contextual equivalence on configuration structures}
\label{sec:def-bisim-cs}

In this section we introduce a notion of context for the configurations structures and then adapt the back-and-forth barbed bisimulation to configurations structures (\autoref{bisim-cs}).
We define hereditary history preserving bisimulation and use two families of relations, denoted \(F_i\) and \(B_i\), to inductively approximate the bisimulation (\autoref{Fn-Bn-and-h}).
We use these relations to show that two processes are barbed congruent whenever their denotations are in the HHPB relation (\autoref{main-thm}).
Once the HHPB has been probed to be a congruence (\autoref{prop:HHPB_congr}), one direction is straightforward , whereas the other is more technical and, as in CCS~\cite{Milner1992}, follows by contradiction.
It uses the relations \(\Forw{i}\) and \(\Backw{i}\) (\autoref{ForwBackwDef}) to build contexts that discriminate processes that are not bisimilar.

\subsection{Contexts for configurations structures}
\label{sec:context_stfam}
Contexts for configurations structures have never been defined as it is not clear what a configuration structure with a hole could be.
However, if a structure \(\conf\) has an operational meaning, \ie there exists \(P\) a process such that \(\conf=\enc{P}\), we use a CCS context \(C[\cdot]\) to build a configuration structure \(\enc{C[P]}\).

When analysing the reductions of a process in context, we need to know the contribution the process and the context have in the reduction.
To this aim we associate to the context \(C[\cdot]\) instantiated by a process \(P\) a projection morphism \(\pi_{C, P}:\enc{C[P]}\to\enc{P}\) that retrieve in \(\enc{C[P]}\) the parts of a configuration belonging to \(\enc{P}\).

Following \autoref{prop:context_that_counts}, we continue to consider only context made of parallel compositions, but the following definition can be extended to arbitrary contexts~\cite[Definition 46]{Cristescu2015}.

\begin{definition}
	\label{def:conf_context}
	Let \(C[\cdot]\) a context, and \(P\) a process.
	The \emph{projection} \(\pi_{C, P}:\enc{C[P]}\to\enc{P}\) is defined on the structure of \(C\) as follows:
	\begin{itemize}
		\item if \(C[\cdot]=C'[\cdot]\mid P'\) then \(\pi_{C, P}:\enc{C'[P]\mid P'}\to\enc{P}\) is defined as \(\pi_{C, P}(e)=\pi_{C', P}(\pi_1(e))\), where \(\pi_1:\enc{C'[P]\mid P'}\to\enc{C'[P]}\) is the projection morphism defined by the product in \autoref{cat-op-def};
		\item if \(C[\cdot]=[\cdot]\) then \(\pi_{C, P}:\enc{C[P]}\to\enc{P}\) is the identity.
	\end{itemize}
\end{definition}
We naturally extend \(\pi_{C, P}\) to configurations, and prove by case analysis that \(\pi_{C, P}:\enc{C[P]}\to\enc{P}\) is a morphism.

\subsection{Relation induced by barbed congruence on configurations structures}
\label{sec:barbed_congruence_stfam}
We define a relation on configurations structures that have an operational meaning %
and we show it is the relation induced by the barbed congruence in RCCS (\autoref{def:sbfc_rccs}).
We call the relation barbed back-and-forth congruence, to highlight its meaning, though it is not strictly speaking a congruence on configurations structures.

\begin{definition}[Back-and-forth barbed congruence on configurations structures]
	\label{bisim-cs}
	A \emph{back-and-forth barbed bisimulation on configurations structures} is a symmetric relation \(\rel\subseteq C_1\times C_2\) such that \((\emptyset,\emptyset)\in {\rel}\), and if \((x_1,x_2)\in {\rel}\), then
	\begin{align*}
		x_1\revredl{e_1}x_1'         & \implies \begin{multlined}[t]
		\exists x_2'\in C_2\text{ \st~} x_2\revredl{e_2}x_2',\\
		\text{with }\labl_1(e_1)=\labl_2(e_2)=\tau\text{ and } (x_1',x_2')\in {\rel};
		\end{multlined}
		\tag{back}\\
		x_1\redl{e_1}x_1'            & \implies \begin{multlined}[t] \exists x_2'\in C_2 \text{ \st~}x_2\redl{e_2}x_2',                                     \\
		\text{with }\labl_1(e_1)=\labl_2(e_2)=\tau \text{ and }(x_1',x_2')\in {\rel};
		\end{multlined}
		\tag{forth} \\
		\text{if }\exists e_1\in E_1 & \text{ \st~}\labl_1(e_1)\neq\tau \begin{multlined}[t]\text{ and }x_1\redl{e_1}x_1' \text{ then } \exists x_2'\in C_2 \\
		\text{\st~}x_2\redl{e_2}x_2'\text{, with }\labl_1(e_1)=\labl_2(e_2).\end{multlined}\tag{barbed}
	\end{align*}
	Let \(\conf_1\bfbisim\conf_2\) if and only if there exists a back-and-forth barbed bisimulation between \(\conf_1\) and \(\conf_2\).

	Define \(\bfcong\) the \emph{back-and-forth barbed congruence} induced on configurations structures as a symmetric relation on configurations structures that have an operational meaning such that
	\[\enc{P_1}\bfcong\enc{P_2} \iff\forall C, \enc{C[P_1]}\bfbisim\enc{C[P_2]}.\]
\end{definition}

We now prove that this relation is the relation induced on the encoding of processes by the barbed back-and-forth congruence (\autoref{def:sbfc_rccs}).
We begin by proving it in the non-contextual case.
We remind the reader that, as we are going to manipulate encoding of RCCS terms, some restrictions on the terms applies (\autoref{rem-concur}).
\begin{proposition}
	\label{prop:bfbisim_corresp}
	For all \(P\) and \(Q\), \(\emptymem \rhd P \bfbisim\emptymem \rhd Q\iff \enc{P}\bfbisim\enc{Q}\).
\end{proposition}
\begin{proof}
	\begin{itemize}
		\item \(\emptymem \rhd P \bfbisim\emptymem \rhd Q\implies \enc{P}\bfbisim\enc{Q}\).
		      Let \({\relCCS}\) be a back-and-forth barbed bisimulation between \(P\) and \(Q\).
		      We show that the following relation
		      \[
		      	{\rel} =
		      	\begin{multlined}[t]
		      		\{(x_1,x_2) \setst x_1\in\enc{P}, x_2\in\enc{Q}, \exists R, S\text{ \st~}\orig{R}=P,\\
		      		\orig S = Q, R \relCCS S\text{ and }\enc{R}=(\enc{P},x_1), \enc{S}=(\enc{Q},x_2)\}
		      	\end{multlined}
		      \]
		      is a back-and-forth barbed bisimulation between \(\enc{P}\) and \(\enc{Q}\).

		      We have that \((\emptyset,\emptyset)\in {\rel}\), let \((x_1,x_2)\in {\rel}\).
		      We have to show that the conditions in \autoref{bisim-cs} hold.
		      Suppose that \(x_1\redl{e_1}x_1'=x_1\cup\{e_1\}\), for \(\labl(e_1)=\tau\).
		      \begin{align*}
		      	x_1\redl{e_1}x_1' & \implies (\enc{P},x_1)\redl{\labl(e_1)}(\enc{P}, x_1') \tag{From \autoref{def:lts_stfam}}                           \\
		      	                  & \implies R\fwlts{i}{\labl(e_1)}R'\text{ s.t. }\enc{R'}=(\enc{P}, x_1') \tag{From \autoref{lem:operational_corresp}} \\
		      	                  & \implies S\fwlts{i'}{\tau}S' \tag{From \(R \relCCS S\)}                                                             \\
		      	                  & \implies (\enc{Q},x_2)\redl{\labl(e_2)}(\enc{Q},x_2') \tag{From \autoref{lem:operational_corresp}}
		      \end{align*}
		      with \(\labl(e_2)=\tau\).
		      We have then \((x_1',x_2')\in {\rel}\).

		      We proceed in a similar manner to show that conditions on the backward transitions and on the barbs hold.

		\item \(\enc{P}\bfbisim\enc{Q}\implies\emptymem\rhd P\bfbisim\emptymem\rhd Q\).
		      Let \({\relCONF}\) be a back-and-forth barbed bisimulation between \(\enc{P}\) and \(\enc{Q}\).
		      We show that the following relation
		      \begin{align*}
		      	{\rel} =
		      	\begin{multlined}[t]
		      	\{(R,S) \setst \erase{\orig{R}}=P, \erase{\orig{S}} = Q\text{ and }\enc{R}=(\enc{P},x_1), \\
		      	\enc{S}=(\enc{Q},x_2),\text{ with }(x_1,x_2)\in {\relCONF} \}
		      	\end{multlined}
		      \end{align*}
		      is a back-and-forth barbed bisimulation between \(P\) and \(Q\).
		      Let \((R,S)\in {\rel}\), the following holds:
		      \begin{align*}
		      	R\fwlts{i}{\tau}R' & \implies(\enc{P},x_1)\redl{\labl(e_1)}(\enc{P}, x_1') \tag{From \autoref{lem:operational_corresp}} \\
		      	                   & \implies(\enc{Q},x_2)\redl{\labl(e_2)}(\enc{Q},x_2') \tag{From \((x_1,x_2)\in {\relCONF}\)}        \\
		      	                   & \implies S\fwlts{i'}{\tau}S' \tag{From \autoref{lem:operational_corresp}}
		      \end{align*}
		      where \(x_1'=x_1\cup\{e_1\}\), \(x_2'=x_2\cup\{e_2\}\) and \(\labl(e_1)=\labl(e_2)=\tau\).
		      We have that
		      \(\orig {R'} = P\), \(\orig {S'} = Q\), \(\enc{R'}=(\enc{P},x_1')\), \(\enc{S'}=(\enc{Q},x_2')\) and \((x_1',x_2')\in {\relCONF}\).
		      Hence \((R',S')\in {\rel}\).

		      To prove that the remaining conditions on the pair \((R,S)\) holds as well is similar.
	\end{itemize}
\end{proof}

The contextual version of the proposition is straightforward.

\begin{lemma}
	\label{soundness-bisim}
	For all singly labelled processes \(R\) and \(S\), \(\orig{R}\sbfbc\orig{S}\iff\encc{\erase{\orig{R}}}\sbfbc\encc{\erase{\orig{S}}}\).
\end{lemma}
\begin{proof}
	\begin{align*}
		\orig{R}\sbfbc\orig{S} & \iff\forall C[\cdot], \context[\orig{R}]\bfbisim \context[\orig{S}] \tag{From \autoref{def:sbfc_rccs}}                                    \\
		                       & \iff \forall C[\cdot],\emptymem \rhd C[\erase{\orig{R}}]\bfbisim \emptymem \rhd C[\erase{\orig{S}}] \tag{From \autoref{def:rccs_context}} \\
		                       & \iff \forall C[\cdot], \enc{C[\erase{\orig{R}}]} \bfbisim \enc{C[\erase{\orig{S}}]} \tag{From \autoref{prop:bfbisim_corresp}}             \\
		                       & \iff \encc{\erase{\orig{R}}}\sbfbc\encc{\erase{\orig{S}}} \tag{From \autoref{bisim-cs}}
	\end{align*}
\end{proof}

\subsection{Inductive characterisation of HHPB}
\label{sec:inductive_hhpb}
Similarly to the proof in CCS, the correspondence between a contextual equivalence and a non-contextual one necessitates to approximate HHPB with (a family of) inductive relations defined on configuration structures.
If we are interested only in the forward direction (as in CCS), the inductive reasoning starts with the empty set, and constructs the bisimilarity relation by adding pairs of configurations reachable in the same manner from the empty set.
However, to approximate HHPB, we need to have an inductive reasoning on the backward transition as well (\autoref{ForwBackwDef}).
These relations are of major importance to prove our main theorem (\autoref{main-thm}), as they re-introduce the possibility of an inductive reasoning thanks to a stratification of the HHPB relation.

\begin{definition} [Hereditary history preserving bisimilarity]
	\label{sbisimilarity_h-def}
	The hereditary history preserving bisimilarity, denoted \(\hhpb\), is the union of all HHPB relations (\autoref{def:hhpb}).
\end{definition}

\begin{figure}
	\centering
	\begin{tikzpicture}
		\node (emptyset1) at (0, -1) {\(\emptyset\)};
		\node (a1) at (-1, 0) {\(\{e_1\}\)};
		\node (b1) at (1, 0) {\(\{e_1'\}\)};
		\draw [->] (emptyset1) -- (a1);
		\draw [->] (emptyset1) -- (b1);
		\node[align=center] (labels) at (0, -3) {\(\labl_1(e_1) =\labl_1(e'_1) = a\)};
		\node (emptyset2) at (4, -1) {\(\emptyset\)};
		\node (a2) at (3, 0) {\(\{e_2\}\)};
		\node (b2) at (5, 0) {\(\{e_2'\}\)};
		\draw [->] (emptyset2) -- (a2);
		\draw [->] (emptyset2) -- (b2);
		\node[align=center] (labels) at (4, -3) {\(\labl_2(e_2) = \labl_2(e_2') =a\)};
		\draw[<->,colorp2, thick] (a1) .. controls (-1, -2.5) and (5, -2.5) .. (b2);
		\draw[<->,colorp2, thick] (b1) .. controls (1, -2) and (3, -2) .. (a2);
		\node[colorp2] (f2) at (2, -1.8) {\(f_2\)};
		\draw[<->,colorp3, thick] (a1) .. controls (-1, 2) and (3, 2) .. (a2);
		\draw[<->,colorp3, thick] (b1) .. controls (1, 2) and (5, 2) .. (b2);
		\node[colorp3] (f1) at (2, 1.8) {\(f_1\)};
	\end{tikzpicture}
	\caption{Two possible hereditary history preserving bisimulations}
	\label{fig:bisimilarity}
\end{figure}

\begin{remark}[On the uniqueness of hereditary history preserving bisimilarity]
	Writing \(\conf_1\hhpb\conf_2\) is an abuse of notation as hereditary history preserving \emph{bisimulations} are defined on \(C_1\times C_2\times\power(E_1\times E_2)\).
	Also the union of all bisimulations may contain triples that do not have \enquote{compatible} bijections.
	For instance, we have two possible bisimulations between the configurations structures of \autoref{fig:bisimilarity}:
	\[f_1 = \{ e_1\leftrightarrow e_2, e_1'\leftrightarrow e_2' \} \qquad f_2 = \{e_1\leftrightarrow e_2', e_1'\leftrightarrow e_2\}\]
	However, the bisimilarity relation contains both tuples \( (\{ e_1, e_2\}, \{ e_1', e_2'\}, f_1) \) and \( (\{ e_1, e_2\}, \{ e_1', e_2'\}, f_2) \).

\end{remark}

We give an inductive characterisation of HHPB by reasoning on the structures up to a level: we ignore the configurations that have greater cardinality then the considered level.
HHPB is then the relation obtained when we reach the top level.
Hence we are able to detect, whenever two configurations structures are not HHPB, at which level the bisimulation does no longer hold.

In the following, denote \(\card(x)\) the cardinality of a set \(x\).
\begin{definition}[Maximal and top configurations]
	A configuration \(x\in\conf\) is \emph{maximal} if there is no configuration \(y\in \conf\) such that \(x\subsetneq y\).
	If moreover \(\forall y\in\conf\), \(\card(y)\leqslant\card(x)\) then \(x\) is a \emph{top configuration}.
\end{definition}

\begin{definition}[\(\Forw{i}\), \(\Backw{i}\)]
	\label{ForwBackwDef}
	Given \(\conf_1\), \(\conf_2\) two configurations structures
	define, for all \(x_1 \in C_1\), \(x_2 \in C_2\) and \(f\) a label and order-preserving function:
	\begin{align*}
		(x_1, x_2, f)\in \Forw{i}                                                  & \iff
		\begin{cases}
		\card(x_1) = \card(x_2) = i,                                               & \text{if }x_1\text{ and }x_2\text{ are maximal } \\
		\forall x'_1, \exists x'_2, x_1 \redl{e_1} x'_1, x_2 \redl{e_1} x'_2 \text{ and }\\
		f=f'\restr x_1\text{ \st~} (x'_1, x'_2, f') \in \Forw{i+1}                 & \text{otherwise}
		\end{cases}\\
		(x_1, x_2, f) \in \Backw{i}                                                & \iff
		\begin{cases}
		(x_1, x_2, f) \in \Forw{i}                                                 & \text{if \(i = 0\)}                              \\
		\forall x'_1, \exists x'_2, x_1 \revredl{e_1} x'_1, x_2 \revredl{e_1} x'_2 \text{ and }\\
		f'=f\restr x_2\text{ \st~}(x'_1, x'_2, f') \in \Forw{i-1} \cap \Backw{i-1} & \text{otherwise}
		\end{cases}
	\end{align*}

\end{definition}
The label and order-preserving function in the two relations helps to ensure that configurations that are in relation have the same labels and causal structure.

The relation \(\Backw{i}\) is built on top of \(\Forw{i}\): it tests for the backward steps all the couples that passed the forward test.
It should be remarked that, with this definition, \(\Backw{i} \subseteq \Forw{i}\), but, at the price of slight modifications, one could have defined \(\Forw{i}\) on top of \(\Backw{i}\).

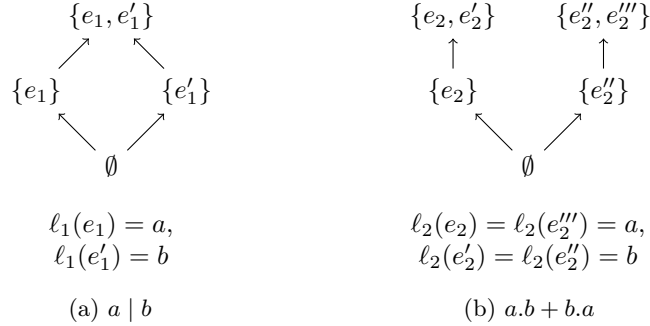
\begin{figure}
	\begin{minipage}[b]{.45\linewidth}
		\centering
		\begin{tikzpicture}
			\node (emptyset) at (0, -1) {\(\emptyset\)};
			\node (a) at (-1, 0) {\(\{e_1\}\)};
			\node (b) at (1, 0) {\(\{e_1'\}\)};
			\node (ab) at (0, 1) {\(\{e_1, e_1'\}\)};
			\draw [->] (emptyset) -- (a);
			\draw [->] (emptyset) -- (b);
			\draw [->] (a) -- (ab);
			\draw [->] (b) -- (ab);
			\node[align=center] (labels) at (0, -2) {\(\labl_1(e_1) = a\),\\ \(\labl_1(e'_1) = b\)};
		\end{tikzpicture}
		\subcaption{\(a \mid b\)}\label{fig:ex_bf_a}
	\end{minipage}
	\begin{minipage}[b]{.45\linewidth}
		\centering
		\begin{tikzpicture}
			\node (emptyset) at (0, -1) {\(\emptyset\)};
			\node (a) at (-1, 0) {\(\{e_2\}\)};
			\node (b) at (1, 0) {\(\{e''_2\}\)};
			\node (ab) at (-1, 1) {\(\{e_2, e'_2\}\)};
			\node (ba) at (1, 1) {\(\{e''_2, e'''_2\}\)};
			\draw [->] (emptyset) -- (a);
			\draw [->] (emptyset) -- (b);
			\draw [->] (a) -- (ab);
			\draw [->] (b) -- (ba);
			\node[align=center] (labels) at (0, -2) {\(\labl_2(e_2) = \labl_2(e_2''') =a\),\\ \(\labl_2(e_2') = \labl_2(e_2'') =b\)};
		\end{tikzpicture}
		\subcaption{\(a.b+b.a\)}\label{fig:ex_bf_b}
	\end{minipage}
	\caption{Encoding parallel and sum in configurations structures}
	\label{fig:ex_bf}
\end{figure}

\begin{example}
	\label{example2}
	Consider the configurations structures in Figures \ref{ex_unif_b} and \ref{ex_unif_c}, the relations \(F_n\) are enough to discriminate them:
	\begin{align*}
		  & F_2 = \big(\{e_1, e_1'\}, \{e_2, e_2'\}\big); \big(\{e_1, e_1'\}, \{e_2'',e_2'''\}\big) \\
		  & \qquad F_1 = \big(\{e_1\}, \{e_2\}\big); \big(\{e_1\}, \{e_2''\}\big)                   \\
		  & \qquad \qquad F_0 = \emptyset
	\end{align*}

	This intuitively is due to the fact that forward transitions are enough to discriminate \(a+a.b\) and \(a.b+a.b\).
	However for comparing the processes \(a \mid b\) and \(a.b+b.a\) whose configurations are in Figures \ref{fig:ex_bf_a} and \ref{fig:ex_bf_b}, we need the backward moves as well.
	Let us first build the \(F_n\) relations:
	\begin{align*}
		  & F_2 = \big(\{e_1,e_1'\},\{e_2,e_2'\}\big); \big(\{e_1,e_1'\};\{e_2'',e_2'''\}\big) \\
		  & \qquad	F_1 = \big(\{e_1\},\{e_2\}\big); \big(\{e_1'\};\{e_2''\}\big)               \\
		  & \qquad \qquad F_0 = \big(\emptyset,\emptyset\big)
	\end{align*}

	We first construct the \(B_0\) relation and then move up in the structures.
	In our example, the \(B_2\) relation breaks the HHPB.
	\begin{align*}
		  & \qquad \qquad B_0 =F_0=\big(\emptyset,\emptyset\big)                 \\
		  & \qquad B_1 = \big(\{e_1\},\{e_2\}\big); \big(\{e_1'\};\{e_2''\}\big) \\
		  & B_2 = \emptyset
	\end{align*}

\end{example}
The following proposition states that pairs of configurations are in a bisimulation relation if they have the same cardinality.
It follows from the fact that any configuration is reachable from the empty set and that they have to mimick each other's step in the backward direction.
\begin{proposition}
	\label{prop:size_config}
	Let \(\conf_1\hhpb\conf_2\) be two configuration structures in a hereditary history preserving bisimulation and \(x_1\in\conf_1\), \(x_2\in\conf_2\) be two configurations.

	If \(\exists f\) such that \((x_1,x_2,f)\in \{\hhpb\}\) then \(\card(x_1)=\card(x_2)\).
\end{proposition}
\begin{proof}
	It follows by induction on the trace \(\emptyset\red^\star x_1\).
	For every event in \(x_1\), we have to add an event in \(x_2\) in order to obtain that the pair \(x_1\) and \(x_2\) are in a HHPB relation.
\end{proof}

The following lemma, that will be handy to prove \autoref{main-thm}, implies that if for all \(n \leqslant k\) the maximum cardinal considered, \(\Forw{n} \cap \Backw{n} \neq \emptyset\), then \(\cup_{n \leqslant k} (\Forw{n} \cap \Backw{n})\) is a bisimulation.

\begin{lemma}
	\label{Fn-Bn-and-h}
	For all \(\conf_1\), \(\conf_2\), if \(\conf_1 \hhpb \conf_2\) , then \(\forall x_1 \in \conf_1 (\exists x_2 \in \conf_2, \exists f, (x_1, x_2, f) \in \Forw{n} \cap \Backw{n}) \iff (\exists x_2 \in \conf_2, \exists f, (x_1, x_2, f) \in\hhpb )\), where \(\card(x_1)=n\).
\end{lemma}

\begin{proof}
	Let us denote \(\rel\) the relation \(\hhpb\).
	One should first remark that \(\conf_1 \hhpb \conf_2\) implies that \(\forall x_1 \in \conf_1, \exists x_2 \in \conf_2\), and \(\exists f\) such that \( (x_1, x_2, f) \in {\rel}\), as \((\emptyset,\emptyset,\emptyset)\in {\rel}\) and all configurations are reachable from the empty set.
	The reader should notice that the \(x_2 \in \conf_2\) and \(f\) on both sides of the \(\iff\) symbols may be different.

	We prove that statement by induction on the cardinal of \(x_1\).
	\paragraph{\(\card(x_1) = 0\)}
	\begin{itemize}
		\item[\(\Rightarrow\)]
		      \(x_2 \in \conf_2\) \st \((\emptyset, x_2, f) \in {\rel}\) follows by the definition of the bisimulation from \(x_2 = \emptyset\) and \(f = \emptyset\).
		\item[\(\Leftarrow\)]
		      By definition, \(\Forw{0} \cap \Backw{0} = \Forw{0}\).
		      Since there exists \(x_2 \in \conf_2\) such that \((\emptyset, x_2, f) \in {\rel}\), we know that any forward transition made by \(\emptyset\) can be simulated by a forward transition from \(x_2\), and that the elements obtained are in the relation \({\rel}\).
		      By an iterated use of this notion, we find top configurations \(x_1^m \in \conf_1\) and \(x_2^m \in \conf_2\) (that is, elements of maximum cardinality, \(k\)) such that \((x_1^m, x_2^m, f^m) \in {\rel}\).
		      By \autoref{prop:size_config}, \(x_1^m\) and \(x_2^m\) have the same cardinality, and \((x_1^m, x_2^m, f^m) \in \Forw{k}\).
		      By just {reversing the trace}, we go backward and stay in relation \(\Forw{i}\) until \(i = 0\), hence we found the \(x_2\) and \(f\) we were looking for.
	\end{itemize}
	\paragraph{\(\card(x_1) = k +1\)}
	As \(\card(x_1) > 0\), we know there exists \(x'_1\) such that \(x_1 \revredl{e_1} x'_1\).
	\begin{itemize}
		\item[\(\Rightarrow\)]
		      Let \(x_2\) and \(f\) such that \((x_1, x_2, f) \in \Forw{k+1} \cap \Backw{k+1}\).
		      We know that
		      \begin{align*}
		      	\forall x'_1, \exists x'_2\text{ and } f', x_1 \revredl{e_1} x'_1, x_2 \revredl{e_1} x'_2 \text{ and } (x'_1, x_2', f') \in \Backw{k} \tag{By Definition of \(\Backw{k} \)} \\
		      	\exists x''_2,f'', (x'_1, x_2'', f'') \in {\rel} \tag{By induction hypothesis}
		      \end{align*}
		      And as \(x'_1 \redl{e_1} x_1\), there exists \(x'''_2\) and \(f'''\) such that \((x_1, x'''_2, f''') \in {\rel}\).
		\item[\(\Leftarrow\)]
		      We prove it by contraposition: suppose that \(\exists x_2,f\) such that \((x_1, x_2, f) \in {\rel}\), we prove that \(\forall x_2\), \((x_1, x_2, f) \notin \Forw{k+1} \cap \Backw{k+1}\) leads to a contradiction.

		      As \((x_1, x_2, f) \in {\rel}\), we know that there exists \(x'_1\) and \(x'_2,f'\) such that \(x_1 \revredl{e_1} x'_1\), \(x_2 \revredl{e_1} x'_2\) and \((x'_1, x'_2, f') \in {\rel}\).
		      By induction hypothesis, \(\exists x''_2\) and \(\exists f''\) such that \((x'_1, x''_2, f'') \in \Forw{k} \cap \Backw{k}\).
		      As \(x'_1 \redl{e_1} x_1\), \(\exists x'''_2\) and \(\exists f'''\) such that \(x''_2 \redl{e_1} x'''_2\) and \((x_1, x'''_2, f''') \in \Forw{k+1}\), by definition of \(\Forw{k}\).

		      So \((x_1, x'''_2, f''') \notin \Backw{k+1}\), but as \(x_1 \revredl{e_1} x'_1\) and \(x'''_2 \revredl{e_1} x''_2\), and as moreover \((x'_1, x''_2, f'') \in \Forw{k} \cap \Backw{k}\), we have that \((x_1, x'''_2, f''') \in \Backw{k+1}\).

		      From this contradiction we know that we found the right element (\(x'''_2\)) that is in relation with \(x_1\) according to \(\Forw{k+1} \cap \Backw{k+1}\).
		      \qedhere
	\end{itemize}
\end{proof}

\begin{figure}
	\begin{tikzpicture}[scale=0.85]
		\fill[colorp1] (0,0) ellipse (1cm and 2cm) node[above=2cm, black]{\(\enc{P_1} = \mem{E_1, C_1, \labl_1}\)};
		\draw (0, -1) node[below]{\(x_1\)} node (x1) {\(\bullet\)};
		\draw (0, 1) node (x'1) {\(\bullet\)};
		\draw [->] (x1) -- node[right]{\(e''_1\)} (x'1);
		\fill[colorp1] (6, 0) ellipse (1cm and 2cm) node[above=2cm, black]{\(\enc{P_1 \mid Q} = \mem{E'_1, C'_1, \labl'_1} = (\enc{P_1} \times \enc{Q}) \restr E_1 \)};
		\fill[colorp1] (6.7, 0) ellipse (1cm and 2cm);
		\draw (6.35, -1) node[below]{\(y_1\)} node (y1) {\(\bullet\)};
		\draw (6.35, 1) node[above]{\(y'_1\)} node (y'1) {\(\bullet\)};
		\draw [->] (y1) -- node[right]{\(e'' = (e''_1, e''_q)\)}(y'1);
		\begin{scope}[yshift=-4.3cm]
			\fill[colorp2] (0,0) ellipse (1cm and 2cm) node[below=2cm, black]{\(\enc{P_2} = \mem{E_2, C_2, \labl_2}\)};
			\draw (0, -1) node[below]{\(x_2\)} node (x2) {\(\bullet\)};
			\draw (0, 1) node[above]{\(x''_2\)} node (x'2) {\(\bullet\)};
			\draw [->] (x2) -- node[right]{\(e''_2\)} (x'2);
			\fill[colorp2] (6, 0) ellipse (1cm and 2cm) node[below=2cm, black]{\(\enc{P_2 \mid Q} = \mem{E'_2, C'_2, \labl'_2} = (\enc{P_2} \times \enc{Q}) \restr E_2\)};
			\fill[colorp2] (6.7, 0) ellipse (1cm and 2cm);
			\draw (6.35, -1) node[below]{\(y_2\)} node (y2) {\(\bullet\)};
			\draw (6.35, 1) node[above]{\(y'_2\)} node (y'2) {\(\bullet\)};
			\draw [->] (y2) -- node[right]{\(e'_2\)}(y'2);
		\end{scope}
		\draw [<->, double] (y2) to [bend left] node[midway, left] {$f_c$} ($(y1)-(.3, 0)$);
		\draw [<->, double] (x2) to [bend left] node[midway, left] {$f$} ($(x1)-(.3, 0)$);
		\draw [->, double] (y2) to [bend right] node[midway, below ,sloped] {$\pi_2$} (x2);
		\draw [->, double] (y1) to [bend right] node[midway, below ,sloped] {$\pi_1$} (x1);
		\draw [->, double] (y'1) to [bend right] node[midway, below ,sloped] {$\pi_1$} (x'1);
		\node[rectangle callout,draw,inner sep=2pt,fill=colorp1,
			callout absolute pointer=(y1.east),
		below right= 25pt and 35pt of y1.north east]
		{\begin{tabular}{c c c c} \(e\) & \( = \)& \(e_1,\) & \(e_q\) \\ \rotatebox{-90}{\(\leqslant\)} & & \rotatebox{-90}{\(\leqslant\)}\(_{\pi_1}\) & \rotatebox{-90}{\(\leqslant\)}\(_{\pi_2}\) \\ \(e'\) & \( =\) & \(e'_1,\) & \(e'_q\)\end{tabular}};
	\end{tikzpicture}
	\caption{Configurations Structures by the end of the proof of \autoref{prop:HHPB_congr}}
	\label{fig-HHPB_congr}
\end{figure}

\subsection{Contextual characterisation of HHPB}
\label{sec:contextual_hhpb}

\begin{proposition}[Hereditary history preserving bisimulation is a congruence]
	\label{prop:HHPB_congr}
	For all singly labelled \(P_1\), \(P_2\), \(\enc{P_1} \hhpb \enc{P_2} \implies \forall C, \enc{C[P_1]} \hhpb \enc{C[P_2]}.\)
\end{proposition}

\begin{proof}
	The proof amounts to carefully build a relation between \(\enc{C[P_1]}\) and \(\enc{C[P_2]}\) that reflects the known bisimulation between \(\enc{P_1}\) and \(\enc{P_2}\).
	Its uses that causality in a product is the result of the entanglement of the causality of its elements (\autoref{prop:cause_projection}).

	Due to the restriction on the contexts we consider, we only have to prove that
	\[
		\forall P_1, P_2, \enc{P_1}\hhpb\enc{P_2}\implies\forall Q, \enc{P_1\vert Q}\hhpb\enc{P_2\vert Q}
	\]

	As \(\enc{P_1}\hhpb\enc{P_2}\), there exists \({\rel}\) a hereditary history preserving bisimulation (HHPB) between \(\enc{P_1}\) and \(\enc{P_2}\).
	\autoref{fig-HHPB_congr} introduces the variables names and types.

	Define \(\rel_c\subseteq C'_1\times C'_2\times\power(E'_1\times E'_2)\) as follows:
	\[
		(y_1, y_2, f_c) \in {\rel}_c \iff
		\begin{cases}
			(\pi_1(y_1),\pi_2(y_2),\pi_1\circ f) \in {\rel}                     \\
			f_c(e)=(\pi_1 \circ f(e)),\pi_2(e))\in y_2\text{ for all }e \in y_1
		\end{cases}
	\]

	Informally \((y_1, y_2, f_c)\) is in the relation \(\rel_c\) if there is \((x_1, x_2, f)\) in \(\rel\) such that \(x_i\) is the first projection of \(y_i\) and such that \(f_c\) satisfies the property: for \((e_1, e_{q})\in E'_1\), \(f_c(e_1,e_{q})=(f(e_1),e_{q})\) and \((f(e_1),e_{q})\in E'_2\).

	Let us show that \(\rel_c\) is a HHPB between \(\mem{E'_1, C'_1, \labl'_1}\) and \(\mem{E'_2, C'_2, \labl'_2}\).
	\begin{itemize}
		\item \((\emptyset,\emptyset,\emptyset)\in {\rel}_c\);
		\item For \((y_1, y_2, f_c)\in {\rel}\) we show that \(f_c\) is label and order preserving bijection.
		      We have that \(f_c\) is defined as \(f_c(e)=(\pi_1 \circ f(e)),\pi_2(e))\), for some \(f\) label and order preserving bijection such that \((\pi_1(y_1),\pi_2(y_2),\pi_1\circ f)\in {\rel}\).

		      That \(f_c\) is a bijection follows from \(f\) being a bijection.

		      Let \(e \in y_1\) with \(\pi_1(e)=e_1\), \(\pi_2(e)=e_{q}\), then \(f_c(e)=(f(e_1),e_{q})\) for some \(f_c\) \st \((\pi(y_1), \pi_2(y_2), f)\in {\rel}\).
		      We have that \(\labl'_1(e)=(\labl_1(e_1),\labl_{Q}(e_{q}))\) and
		      \[
		      	\labl'_2(f_c(e))=\labl'_2(f(e_1),e_{q})=\big(\labl_2(f(e_1)),\labl_{Q}(e_{q})\big)
		      \]
		      As \(f\) is label preserving we get \(\labl'_2(f_c(e))=(\labl_1(e_1),\labl_{Q}(e_{q}))\), hence \(\labl'_1(e)=\labl'_2(f_c(e))\).

		      Let us now show that for \(e, e' \in y_1\), if \(e\to_{y_1} e'\) then \(f_c(e)\leqslant_{y_2} f_c(e')\).
		      We denote \(\pi_1(e)=e_1\), \(\pi_2(e)=e_{q}\) and \(\pi_1(e')=e_1'\), \(\pi_2(e')=e_{q}'\).
		      Then from \autoref{prop:cause_projection}
		      \begin{equation*}
		      	e\to_{y_1} e'\implies e_1\leqslant_{\pi_1(y_1)} e_1'\text{ or }e_{q}\leqslant_{\pi_2(y_1)} e_{q}'
		      \end{equation*}
		      We consider the case where \(e_1\leqslant_{\pi_1(y_1)} e_1'\).
		      As \(f\) is order preserving we have that \(f(e_1)\leqslant_{\pi_1(y_2)}f(e_1')\).
		      Then \((f(e_1),e_{q})\leqslant_{x_2}(f(e_1'),e_{q}')\), as the projections are order reflecting.

		\item Let \((y_1, y_2,f_c)\in {\rel}_c\) and \(y_1 \redl{e''} y_1'\), \(y_1'=y_1\cup\{e''\}\).
		      We consider only the case when \(\pi_1(e'') = e''_1\neq\star\), \(\pi_2(e'')= e''_{q}\neq\star\) as the rest is similar.
		      From the definition of the projections \(\pi_1(y_1)\), \(\pi_1(y_1')\in C_1'\) and as \(\pi_1(e'')=e''_1\neq\star\), we have that \(\pi_1(y_1')=\pi_1(y_1)\cup\{e''_1\}\).
		      We reason similarly on \(\pi_2(y_1)\) and get
		      \begin{equation}
		      	\label{eq1}
		      	\pi_1(y_1)\redl{e''_1}\pi_1(y_1')\text{ and }\pi_2(y_1)\redl{e''_{q}}\pi_2(y_1').
		      \end{equation}
		      From \autoref{eq1} and as \((\pi_1(y_1), \pi_2(y_2),f)\in {\rel}\), by definition of \(\rel_c\), we have that
		      \begin{equation}
		      	\label{eq2}
		      	\exists x_2'\text{ \st~}\pi_1(y_2)\redl{e''_2}x_2'=x_2\cup\{e''_2\}
		      \end{equation}
		      and
		      \begin{equation}\label{eq3}
		      	f'=f\cup\{e_1''\leftrightarrow e_2''\}
		      \end{equation}
		      such that \((x_1',x_2',f')\in {\rel}\).

		      Let us show that \(\exists y'_2 \in (\enc{P_2}\times\enc{P_{Q}})\) with \(y_2'=y_2\cup\{e_2'\}\) and \(\pi_1(e'_2)=e''_2\), \(\pi_2(e'_2)=e_{q}''\).
		      From \autoref{eq1} and \autoref{eq2} we have that the projections are defines with \(\pi_1(y'_2)=x'_2\), \(\pi_2(y'_2)=\pi_2(y_1')\).
		      The axioms of finiteness and coincidence freeness on \(y_2'\) follows from \(y_2\) being a configuration in \((\enc{P_2}\times\enc{P_{Q}})\).

		      Let us show that \(y_2'\notin X_2\).
		      We have that \(y_1'\notin X_1\).
		      As \(\labl(e_1'')\) and \(\labl(e_q'')\) are compatible, then so are \(\labl(e_2'')\) and \(\labl(e_q'')\), hence \(y_2\cup\{(e_2'',e_q'')\}\notin X_2\).

		      Remains to show \((y_1',y_2',f_c')\in {\rel}\), where \(f_c'=f_c\cup\{e''_1 \leftrightarrow e''_2\}\).
		      We have that \((\pi_1(y_1'),\pi_1(y_2'),f')\in {\rel}_c\) and from \autoref{eq3} that \(\pi_1\circ f_c'=f'\). \qedhere
	\end{itemize}
\end{proof}

\newcommand{\horiz}{4.6} %
\pgfmathsetmacro{\horizhoriz}{2*\horiz}

\begin{figure}
	For \(i \in \{1, 2\}\), we have:
	\begin{tikzpicture}
		\fill[colorp1] (0,0) ellipse (1cm and 2cm) node[below=1.5cm, left=.7cm, black]{\(\enc{P_i}\)};
		\draw (0, -1) node[below]{\(x_i\)} node (x1) {\(\bullet\)};
		\draw (0, 1) node[above]{\(x'_i\)} node (x'1) {\(\bullet\)};
		\draw [->] (x1) -- (x'1);
		\fill[colorp1] (\horiz,0) ellipse (1cm and 2cm) node[below=1.5cm, left=.7cm, black]{\(\enc{C[P_i]}\)};
		\draw (\horiz, -1) node[below]{\(y_i\)} node (y1) {\(\bullet\)};
		\draw (\horiz, 1) node[above]{\(y'_i\)} node (y'1) {\(\bullet\)};
		\draw [->] (y1) -- (y'1);
		\fill[colorp1] (\horizhoriz,0) ellipse (1cm and 2cm) node[below=1.5cm, left=.7cm, black]{\(\enc{C'[P_i]}\)};
		\draw (\horizhoriz, 1) node[above]{\(z'_i\)} node (z'1) {\(\bullet\)};
		\draw [->, double] (y1) to [bend right] node[midway,below,sloped] {\(\pi_{C, P_i}\)} (x1);
		\draw [->, double] (z'1) to [bend right] node[midway,below,sloped] {\(\pi_{C', C[P_i]}\)} (y'1);
		\draw [->, double] (y'1) to [bend right] node[midway,below,sloped] {\(\pi_{C, P_i}\)} (x'1);
	\end{tikzpicture}

	We start with \(y_1 \sbfbc y_2\), then prove that \(z'_1 \sbfbc z'_2\), to end up with \((x'_1, x'_2, f) \in \Forw{n} \cap \Backw{n}\).

	\caption{Configurations Structures by the end of the proof of \autoref{main-thm}}
	\label{fig-main-thm}
\end{figure}

\begin{theorem}
	\label{main-thm}
	For all singly labelled \(P_1\) and \(P_2\), \(\enc{P_1} \hhpb \enc{P_2} \iff \enc{P_1} \sbfbc \enc{P_2}\).
\end{theorem}

\begin{proof}
	The left-to-right direction follows from the definition of \(\hhpb\) (\autoref{def:hhpb}) and from \autoref{prop:HHPB_congr}.

	We prove the other direction %
	by contraposition: let us suppose that \(\enc{P_1}\sbfbc\enc{P_2}\) and \(\enc{P_1} \not\hhpb \enc{P_2}\), we will find a contradiction.
	\autoref{fig-main-thm} presents the general shape of the configurations at the end of the proof.

	As \(\enc{P_1} \not\hhpb \enc{P_2}\), by \autoref{Fn-Bn-and-h}, there exists \(x_1\in\encc{P_1}\) such that \(\forall x_2\in\encc{P_2}\), \((x_1,x_2,f)\notin F_n\cap B_n\) holds.
	Let us consider the largest such \(x_1\).
	Note that we consider only \(x_2\) such that \(\card(x_1)=\card(x_2)=n\), and that we use the projections \(\pi_{C, P}\) (\autoref{def:conf_context}) to separate the events of the process \(P\) from the events of the context \(C\).

	For any \(x_1\) we define
	\( C[\cdot] \coloneqq \prod_{e_i \in x_i} (\overline{\labl(e_i)} + c_{e_i}) \mid [\cdot] \)
	where \(c_{e_i}\notin \nm{P_1}\cup\nm{P_2}\), such that the following holds
	\begin{itemize}
		\item \(\exists y_1 \in \enc{C[P_1]}\) such that \(y_1\) is closed, \(\pi_{C, P_1} (y_1) = x_1\) and \(y_1 \not\downarrow_{c_{e_i}}\) for all \(e_i \in x_1\);
		\item We supposed that \(\enc{P_1}\sbfbc \enc{P_2}\), so \(\enc{C[P_1]}{\bfbisim}\enc{C[P_2]}\).
		      Hence \(\exists\rel\) a back-and-forth barbed bisimulation and \(\exists y_2 \in \enc{C[P_2]}\) such that \((y_1, y_2) \in {\rel}\) and \(y_2 \not\downarrow_{c_{e_i}}\) for all \(e_i \in x_1\).
	\end{itemize}
	We proceed as follows:
	\begin{itemize}
		\item we show that there exists \(f\) a label and order preserving bijection between \(x_1\) and \(\pi_{C, P_1}(y_2)\);
		\item then we show that \((x_1,\pi_{C, P_1}(y_2),f)\in F_n\) for \(f\) defined above;
		\item similarly we show that \((x_1,\pi_{C, P_1}(y_2),f)\in B_n\).
	\end{itemize}

	We denote \(\pi_{C, P_1}(y_2)\) with \(x_2\).
	We have by induction on the trace \(\emptyset\red^{\star} y_1\) that if \(y_1\) is closed then \(y_2\) is closed as well.
	Moreover we define a bijection \(g:y_1\to y_2\) that is order and label preserving.
	It follows again from an induction on the trace \(\emptyset\red^{\star} y_1\) and from \(y_2 \not\downarrow_{c_{e_i}}\) for all \(e_i \in x_1\).

	We have that \(\forall e_1, e'_1 \in x_1\), and \(e_2 \in x_2\),
	\begin{equation}
		\label{reason2}
		e_2 \in x_2 \iff e_1 \in x_2\text{ and } \labl(e_1) = \labl (e_2)
	\end{equation}
	\begin{align}
		e_1 <_{x_1} e'_1 & \Longrightarrow \pi_{C, P_1}^{-1} (e_1) <_{y_1} \pi_{C, P_1}^{-1} (e'_1) \label{reason}        \\
		                 & \Longrightarrow g(\pi_{C, P_1}^{-1} (e_1)) <_{y_2} g(\pi_{C, P_1}^{-1} (e'_1)) \label{reason3}
	\end{align}

	Remark that \eqref{reason2} follows from \(y_2 \not\downarrow_{c_{e_i}}\) and from the fact that if \(y_1\) is closed we can show by contradiction that \(y_2\) is closed as well.
	Secondly, \eqref{reason} follows from the morphisms reflecting causality.
	Lastly, \eqref{reason3} follows from \(g\) being an order preserving bijection between \(y_1\) and \(y_2\).

	For every events in \(e_1'',e_2''\in y_2\) such that \(e_1''\to_{y_2} e_2''\) from \autoref{prop:cause_projection}, either \(\pi_{C, P_2}(e_1'')\leqslant_{\pi_{C, P_2}(y_2)} \pi_{C, P_2}(e_2'')\) or the projection of the two events are causal dependent in the context.
	However, the context does not induce any causality between the events.
	As \(\pi_{C, P_2}(e_1'')\leqslant_{\pi_{C, P_2}(y_2)} \pi_{C, P_2}(e_2'')\), we have that there exists \(f\) a label and order preserving bijection between \(x_1\) and \(\pi_{C, P_1}(y_2)\).

	Let us now prove that \((x_1, x_2, f) \in \Forw{n+1}\).
	There are two cases:
	\begin{align}
		\not\exists x'_1, x_1 \redl{e_1} x'_1, \exists x'_2, x_2 \redl{e_2} x'_2 \label{case-no-transition}                                         \\
		\exists x'_1, x_1 \redl{e_1} x'_1, \forall x'_2, x_2 \redl{e_2} x'_2 \text{ and } (x'_1, x'_2, f') \notin \Forw{k}\label{case-no-extension}
	\end{align}

	The implication \eqref{case-no-transition} is easier: if \(\exists x'_2, x_2 \redl{e_2} x'_2\), then, as a context cannot remove transitions from the original process, \(\exists y'_2, y_2 \redl{(e_2, \star)} y'_2\).
	As \(\enc{C[P_2]} \bfbisim \enc{C[P_1]}\), \(\exists y'_1, y_1 \redl{(e_1, \star)} y'_1\), and a similar argument on the context shows that \(\exists x'_1, x_1 \redl{e_1} x'_1\).
	Hence a contradiction.

	To prove \eqref{case-no-extension} requires more work.
	First, let
	\(C'[\cdot] \coloneqq C[\cdot] \mid (\overline{\labl(e_1)} + c_{e_1})\).
	By induction hypothesis, there exists \(z'_1 \in \enc{C'[P_1]}\) such that \(z'_1\) is closed, \(\pi_{C', C[P_1]} (z'_1) = y'_1\) and \(z'_1 \not\downarrow_{c_{e_i}}\) and \(z'_1 \not\downarrow_{c_{e_1}}\) for all \(e_i \in x_1\).

	By hypothesis, \(\enc{P_1}\sbfbc \enc{P_2}\), hence there exists \(\rel'\) a back-and-forth barbed bisimulation between \(\enc{C'[P_1]}\) and \(\enc{C'[P_2]}\).
	It implies that \(\exists z_2'\) such that \(z_2 \in \enc{C'[P_2]}\) and \(z'_2 \not\downarrow_{c_{e_i}}\) and \(z'_2 \not\downarrow_{c_{e_1}}\) for all \(e_i \in x_1\).

	Using a similar argument to above we have that \(z_2'\) is closed and that there exists a bijection \(h\) between \(z_1'\) and \(z_2'\).

	Let us denote the projection \(\pi_{C', P_2} (z'_2)\) as \(x_2'\).
	We infer using the fact that \(z_2'\) is closed and that \(z'_2 \not\downarrow_{c_{e_1}}\) that \(\exists e_2'\in x_2'\) such that \(\labl(e_2')=\labl(e_1)\).

	As there exists a label and order preserving bijection \(h'\) between \(z_1'\) and \(z_2'\), and as we forbid auto concurrency and ambiguous non-deterministic sum (\autoref{rem-concur}), we conclude that
	\(x_2'\setminus\{e_2'\}=x_2\), for \(\pi_{C', P_2} (z'_2)=x_2'\).

	Then we have \(\pi_{C', P_1}(z'_1) = x'_1, \pi_{C', P_2}(z'_2) = x'_2\) and \(f\cup\{e_1'\leftrightarrow e_2'\}\) a bijection between the two.
	As we supposed that \(x_1\) is the largest configuration for which the HHPB breaks we get that \(\exists x_2''\) such that \((x_1',x_2'',f'')\in\Forw{n+1}\).
	But such an \(x_2''\) is unique since \(P_2\) is singly labelled.
	Thus we conclude that \( (x'_1, x'_2, f\cup\{e_1'\leftrightarrow e_2'\}) \in \Forw{n+1}\).

	The proof that \((x_1, x_2, f) \in \Backw{n}\) goes along the line of (and uses) the proof that \((x_1, x_2, f) \in \Forw{n}\).
\end{proof}

\begin{remark}[On \autoref{main-thm}]
	Note that \autoref{main-thm} is a result on RCCS processes that have \emph{an empty memory}.
	It is a consequence of HHPB and the back-and-forth barbed congruence on configuration structure (\autoref{bisim-cs}) being defined on configurations structures, and not on the tuples of configurations structures and configurations.
	However, we need the reversible setting to simulate the back-and-forth behaviour that we acquire when moving to configurations structures.
	The result above then should be read as: \emph{reversible process with \emph{an empty memory} are barbed congurent if and only if their encodings in configurations structures are in a HHPB relation}.

	To make the result more general and include any reversible process we need to reformulate it as follows.

	\begin{conjecture}
		If \(R\sbfbc S\) such that \(\encc{R}=(\conf_R, x_R)\) and \(\encc{S}=(\conf_S, x_S)\) then there exists \(\rel\) a HHPB between \(\conf_R\) and \(\conf_S\) with \((x_R,x_S,f)\in {\rel}\), for some \(f\).
	\end{conjecture}
	We leave this as future work.
\end{remark}

\section*{Conclusions and future work}
\addcontentsline{toc}{section}{Conclusions and future work}
We showed that, for a restricted class of RCCS processes (without recursion, auto-concurrency nor auto-conflict (\autoref{rem-concur})) hereditary history preserving bisimilarity has a contextual characterisation in CCS.
We used the barbed congruence defined on RCCS as the congruence of reference, adapted it to configurations structures and then showed a correspondence with HHPB.
As a proof tool, we defined two inductively relations that approximate HHPB.
Consequently we have that adding reversibility into the syntax helps in retrieving some of the discriminating power of configurations structures.

Note that one could prove the main result of the paper by showing that the bisimulation defined on the LTS of RCCS and the barbed congruence (\autoref{def:sbfc_rccs}) equate the same terms.
We chose to use configurations structures instead, as we plan to investigate other equivalences on reversible process algebra and their interpretations in configurations structures give interesting insights.

\paragraph{Weak equivalences}
This work follows notable efforts~\cite{Phillips2007,Lanese2010} to understand equivalences for reversible processes.
There are numerous interesting continuations.
A first one is to move to weak equivalences, which ignores silent moves \(\tau\) and focus on the observable part of a process.
This is arguably a more interesting relation than the strong one, in which processes have to mimick \emph{exactly} each other's silent moves.
Even if such a relation on configurations structures exists~\cite{Vogler1993,Fiore1999} one still has to show that this is indeed the relation we expect.

In configuration structures, the adjective \emph{weak} has sometimes~\cite{Phillips2012,Glabeek1989} a different meaning: it stands for the ability to change the label and order preserving bijection as the relation grows, to modify choices that were made before this step.
It would be interesting to understand what \enquote{weak} relations in this sense represent for reversible processes.

\paragraph{Insensitiveness to the direction of the transitions and irreversibility}
The relations defined so far simulate forward (\resp backward) transitions only with forward (\resp backward) transitions, and only consider \emph{forward} barb.
Ignoring the direction of the transitions could introduce some fruitful liberality in the way processes simulate each other.
Depending on the answer, \(a + \tau . b\) and \(a + b\) would be weakly bisimilar or not.
A weak bisimulation that ignores the direction of transitions~\cite{Lanese2010} already exists, but it equates a reversible process with all its derivatives.
Irreversible moves could play an important role in such equivalences and would help to understand what are the meaningful equivalences in the setting of transactions~\cite{Danos2005}.

Reversibility is commonly used in transactional systems, \ie participative computations where a commitment phase is reached whenever a consensus occurs.
This has two effects: it forbids the further exploration of the solution space, and prevents all the participants to complete if a participant cancels the transaction~\cite{Danos2007}.
Commitment is modelled as an \emph{irreversible} action: such a feature is present in RCCS~\cite{Danos2004}, but absent from our work.
It could probably be implemented by adding a mechanism to \enquote{update} the origin of a term, and by \enquote{cutting} the configuration structure after an irreversible transition (in the spirit of the LTS of \autoref{def:transition_stable}).
However, it remains to prove that those two actions would be equivalent.

\paragraph{Removing the limitations}
Context---which plays a key role---raise questions on the memory handling of RCCS : what about a context that could fix the memory of an incoherent process?

Maybe of less interest but important for the generality of these results, one should include infinite processes as well.
This needs a rework of the relations in \autoref{ForwBackwDef} used to approximate the HHPB.
In configurations structures however one usually handles the recursive case by unfolding the process up to a finite level.

One way to retrieve the class of processes with auto-conflict and auto-concurrence could be to define bisimulations that take into account tagged labels.
At the price of a verbose syntax, one could imagine being able to discriminate between configurations reached after firing events with the same labels, thus allowing to define configurations structures for arbitrary RCCS terms.
Are relations taking into account those \enquote{localities}~\cite{Boudol1988a}, which uniquely determine occurrences of a label, more discriminating the traditional bisimulations?

Lastly, we conjecture that HHPB is equivalent to a congruence relation on terms that do not exhibit auto-conflict.
More precisely, we could imagine that congruent processes have isomorphic event structures, and that configurations structures are isomorphic if and only if they are in HHPB relation.

\section*{Acknowledgement}
We would like to warmly thank D.~Varacca and J.~Krivine for the useful discussions as well as the referees of an earlier version~\cite{Aubert2015d} for their helpful remarks.

\section*{References}
\bibliographystyle{elsarticle-num.bst}
\bibliography{standalone}
\end{document}